\documentclass[11pt]{article}	
\usepackage{amsmath,amssymb,amsthm}   
\usepackage{amsfonts}
\usepackage{mathtools}
\usepackage{breqn}              
\usepackage{bm}
\usepackage{graphicx}
\usepackage{epstopdf}
\usepackage{subcaption}
\usepackage[a4paper, total={6.5in, 9in}]{geometry}
\usepackage{authblk}
\usepackage{color,soul}
\usepackage[english]{babel}
\usepackage{stackengine}
\usepackage{hyperref}

\usepackage{framed} 
\usepackage{multicol} 
 \usepackage{nomencl} 
\renewcommand*\nompreamble{\begin{multicols}{2}}
\renewcommand*\nompostamble{\end{multicols}}
\makenomenclature

\usepackage{natbib}
\setcitestyle{authoryear,\open={(},close={)}}

\usepackage{bigints}
\usepackage[title,titletoc,toc]{appendix}
\usepackage{etoolbox}
\appto\appendices{\counterwithin{equation}{section}}
\usepackage{epsfig}
\newtheorem{prop}{Proposition}[subsection]

\hypersetup{
    colorlinks=true,
    linkcolor=blue,
    filecolor=blue,      
    urlcolor=blue,
    citecolor=blue
}

\AtBeginDocument{%
    \LetLtxMacro\oldref{\ref}%
    \DeclareRobustCommand{\ref}[2][]{{\color{blue}(\oldref#1{#2})}}%
     \LetLtxMacro\oldcite{\cite}%
    \DeclareRobustCommand{\cite}[2][]{{\color{blue}\oldcite#1{#2}}}%
    \LetLtxMacro\oldcitep{\citep}%
    \DeclareRobustCommand{\citep}[2][]{{\color{blue}\oldcitep#1{#2}}}%
}

\newcommand\barbelow[1]{\stackunder[1.2pt]{$#1$}{\rule{1.0ex}{.075ex}}}

\newcommand{\bs}{\boldsymbol}
\newcommand{\tb}{\textbf}
\newcommand{\dbar}{{\mkern3mu\mathchar'26\mkern-12mu d}}

\DeclareRobustCommand{\rchi}{{\mathpalette\irchi\relax}}
\newcommand{\irchi}[2]{\raisebox{\depth}{$#1\chi$}} 

\newcommand\groupequation[2][30pt]{%
  \setbox0=\hbox{$\displaystyle#2$}%
  \stackengine{0pt}{\copy0}{%
    \makebox[\linewidth]{\hfill$\left.\rule{0pt}{\ht0}\right\}$\kern#1}}
    {O}{c}{F}{T}{L}
}

\begin{document}

\title{A space-time gauge theory for modeling ductile damage and its NOSB peridynamic implementation}
\author[1]{Sanjeev Kumar  \thanks{sanjeevkumar@iisc.ac.in}}
\affil[1]{Department of Materials Engineering, Indian Institute of Science, Bangalore 560012, India }

\date{}
\maketitle

\section*{Abstract}

Local translational and scaling symmetries in space-time is exploited for modelling ductile damage in metals and alloys over wide ranges of strain rate and temperature. The invariant energy density corresponding to the ductile deformation is constructed through the gauge invariant curvature tensor by imposing the Weyl like condition. In contrast, the energetics of the plastic deformation is brought in through the gauge compensating field emerged due to local translation and attempted to explore the geometric interpretation of certain internal variables often used in classical viscoplasticity models. Invariance of the energy density under the local action of translation and scaling is preserved through minimally replaced space-time gauge covariant operators. Minimal replacement introduces two non-trivial gauge compensating fields pertaining to local translation and scaling. These are used to describe ductile damage, including plastic flow and micro-crack evolution in the material. A space-time pseudo-Riemannian metric is used to lay out the kinematics in a finite-deformation setting. Recognizing the available insights in classical theories of viscoplasticity, we also establish a correspondence of the gauge compensating field due to spatial translation with Kr\"{o}ner's multiplicative decomposition of the deformation gradient. Thermodynamically consistent coupling between viscoplasticity and ductile damage is ensured through an appropriate degradation function.  Non-ordinary state-based (NOSB) peridynamics (PD) discretization of the model is used for numerical implementation. We conduct simulations of uniaxial deformation to validate the model against available experimental evidence and to assess its predictive features. The model's viability is tested in reproducing a few experimentally known facts, viz., strain rate locking in the stress-strain response, whose origin is traced to a nonlinear microscopic inertia term arising out of the space-time translation symmetry. Finally, we solved 2D and axisymmetric deformation problems for qualitatively validating the model's viability. NOSB peridynamics axisymmetric formulation in finite deformation setup is also presented.

\tb{Key Words:}  Viscoplasticity; Ductile-damage; Gauge theory; Large deformation; Axisymmetry; Peridynamics


\section{Introduction}

As with viscoplasticity, mechanicians have, over the years, sustained an animated research effort at a better understanding of ductile damage, especially from micromechanical or phenomenological perspectives -- via both continuum modelling and experiments conducted at different scales. Large plastic deformation that one typically encounters in several industrial applications -- machining, metal forming (drawing, rolling, extrusion) or fracking under large confining pressure to wit, is often attended to by damage. By the word `damage', one often alludes to an inelastic phenomenon that involves a measure of reduction in the effective cross sectional area around a material point so that the load carrying capacity is reduced. Such a reduction in area could come about through the formation of irreversible microcracks or nucleation and coalescence of voids in the mesoscale. 


The twin processes of damage and plastic deformation result in macroscopic signatures, viz.  stiffness degradation, reduction in effective cross-section that offers resistance against external loads and ductility - perhaps up to a critical state, beyond which rupture may occur. Substantial effort has been invested in unraveling and mathematically describing the physical phenomena underlying the highly nonlinear dynamical processes resulting in ductile damage at the continuum scale. Acceptable predictions of such processes require, among others, a rational coupling of viscoplastic deformation with damage. It also requires an accurate relationship among flow stress, temperature, strain rate and accumulated plastic strain or other internal variables presenting the micro-morphology of polycrystalline solids during plastic deformation and damage. Over the years, many constitutive schemes have been drawn up through phenomenological or physical motivations. For viscoplastic deformation with damage, the phenomenological constitutive model proposed by \cite{johnson1983constitutive} deserves a special mention as it is simple to implement and the material parameters are already established for many metals and alloys of industrial significance. There are however limitations with this approach. For instance, the assumed uncoupling of strain, strain rate, and temperature in the Johnson-Cook (JC) model implies that it is hardly applicable to thermo-viscoplastic materials such as high-strength low-alloy steel (HSLA), Molybdenum etc., where rate sensitivity of flow stress changes with temperature. A few prominent examples of physics-based viscoplastic constitutive models that use the evolution of dislocation density, are by \cite{zerilli1987dislocation, follansbee1988constitutive, estrin1996dislocation, voyiadjis2005microstructural, krasnikov2011dislocation}. Depending on the rate-controlling mechanism specific to the material structure type, Zerilli and Armstrong (ZA) proposed constitutive expressions for different crystalline structures \citep{zerilli1987dislocation}. A modified ZA model in the high-temperature range is proposed by \cite{voyiadjis2005microstructural}, which also include the strain rate effect on the thermal activation area to study plastic deformation in metals under high temperature. \cite{langer2010thermodynamic} has proposed a nonequilibrium  thermodynamic framework based on two temperatures (including an effective temperature) for modeling dislocation mediated plastic flow for face-centered cubic metals. Some other prominent viscoplasticity model could be by \cite{article, khan1999behaviors, gao2012constitutive, knezevic2013polycrystal, kabirian2014negative, li2019machine}. For work on gradient-based viscoplasticity theories, we refer to \cite{muhlhaus1991variational, zbib1992gradient, hutchinson1997strain, gurtin2000plasticity, gurtin2005theory, gurtin2010mechanics}.

Local theories of plasticity coupled with damage often fall short in replicating experimental observations, viz., localization of plastic and damage zones induced due to the excessive softening.  Gradient-based plasticity and damage models overcome this problem by limiting the width of localized plastic and damage bands due to length scales associated with the higher order gradient terms.  
For modeling viscoplastic response in the shock wave regime, viz. under extreme strain rates, a few physically motivated models are reported by \cite{preston2003model, austin2011dislocation, crowhurst2011invariance, ravelo2013shock, mayer2013modeling}. For a few early and popular ductile damage models,  we refer to \cite{gurson1977continuum, tvergaard1981influence, tvergaard1984analysis, lemaitre1985continuous, tvergaard1989material}. A coupled theory of continuum damage mechanics and finite strain plasticity is formulated by \cite{voyiadjis1992plasticity} in the Eulerian reference system. Various aspects of an anisotropic damage model coupled to plasticity is discussed by \cite{al2003coupling}. \cite{brunig2008ductile} investigated the ductile damage criterion at various stress triaxialities and discussed the effect of stress triaxiality on the onset and evolution of damage in ductile metals. \cite{brunig2011modeling} presented the continuum model to investigate ductile damage and fracture behavior based on different micromechanisms. \cite{shojaei2013viscoplastic} presented viscoplastic constitutive theory for brittle to ductile damage in polycrystalline materials under dynamic loading.   A non-associative finite strain anisotropic elastoplastic model fully coupled with anisotropic ductile damage is proposed by \cite{badreddine2015damage}.  This is an area that has, of late, drawn great interest resulting in a significant body of literature on improved models for ductile damage. For few more recent research articles, see \cite{khan2016deformation, gholipour2019experimental, reddi2019ductile, chen2019investigation, sancho2019experimental, ganjiani2020damage, sabik2022tensile, beygi2022utilizing, zhang2022effects, wei2022analysis, de2023multi}. The literature on the ductile fracture is vast. The extensive literature on ductile fracture can be explored further by referring to \cite{besson2010continuum} and \cite{volegov2016damage} for comprehensive reviews on continuum models and experimental studies in this field. Phase-field models, incorporating a regularized Griffith-type \citep{griffith1921vi} brittle fracture approach, represent cracks by means of additional continuous field variables. Diffusive crack representation circumvents certain pathologies of a sharp interface description of cracks. Phase-field based brittle fracture models have been investigated, among others, by \cite{hakim2009laws, kuhn2010continuum, bourdin2011time, borden2012phase, hofacker2012continuum, kumar2018emergence}.  Phase-field based models of ductile fracture, which involve a coupling of damage mechanics with elasto-plasticity, have been presented by \cite{alessi2015gradient,miehe2016phase}. Such a model of ductile fracture  by \cite{ambati2015phase} multiplicatively couples damage with local plasticity through a degradation function.

A rationally grounded continuum model that can reproduce the essential features of ductile damage at the macroscale in good agreement with experimental evidence remains conspicuous by its absence. \cite{kadic1982yang} tried to use the Yang-Mills gauge theory to model plasticity with dislocation and disclination fields characterized through local rotational symmetry SO(3) and translational symmetry T(3).  The study showed that local symmetry breaking of T(3) action can model the evolution of dislocations. In contrast, disclinations and induced (rotational) dislocations are due to the local symmetry breaking of SO(3) group action. \cite{lagoudas1989material, edelen2012gauge}  proposed material and spatial gauge theories of solids to explain the dynamic behaviour of continuously distributed defects by exploring the local action of the material symmetry group $G_m = \left\lbrace SO(3) \triangleright T(3)\right\rbrace \times T(1)$ and the spatial symmetry group $G_s = \left\lbrace SO(3) \triangleright T(3)\right\rbrace$ on Lagrange density. A conformal gauge theory is presented by \cite{kumar2024modelling} for modelling coupled electro-mechanical phenomena in elastic dielectrics using local scaling symmetry. A space–time gauge theory for dynamic plasticity is proposed by \cite{kumar2024space} to investigate the role of non-linear micro-inertia on the stress-strain response of various metals and alloys. In the present study, we report an attempt to construct a rationally grounded continuum framework for modelling ductile damage in metals and alloys using the space-time gauge theory proposed by \cite{roy2019cauchy,roy2020cauchy, kumar2024space}. We exploit the local translational symmetry and the conformal symmetry in space-time to incorporate the history effect vis-\`a-vis ductile damage in our formulation.  Due to the local (inhomogeneous) translation in space and time, the energy density in its classical form will no longer be invariant. To restore invariance, we define the space-time gauge covariant derivatives in place of ordinary partial derivatives through minimal replacement. This procedure introduces gauge compensating one form fields $\bs{\bar{\Phi}}$ and $\bs{\Psi}$ in space and time. Using space-time gauge covariant derivatives and the Minkowski metric, we define the distance measure between two infinitesimally close points in both reference and current configurations. We also derive the 4D pulled-back metric. We may compute the distance between two infinitesimally separated points in terms of the current configuration using the coordinates of the reference configuration and the 4D pulled-back metric. We construct the invariant energy density under the action of local translation in space-time using the components of the 4D pulled-back metric. Next, we derive the relationship between the gauge compensating fields $\bs{\hat{\Phi}}$ and $\bs{F}^{p-1}$ and establish a correspondence with the popular multiplicative decomposition of the deformation gradient, $\bs{F} = \bs{F}^e \bs{F}^p$. Instead of postulating that the covariant derivative of the metric tensor is zero, we imposed a more general Weyl-like condition on the reference metric and construct a space-time gauge-invariant curvature tensor by exploiting the local conformal symmetry. The scalar gauge-invariant curvature is defined by appropriately contracting the space-time gauge-invariant curvature with the metric tensor. The scalar curvature is then employed to construct the invariant energy density associated with ductile damage. Next, the constitutive relations that provide a closure of the governing equations of motion and the thermodynamic restrictions are presented. We also derive the temperature evolution equation employing the laws of thermodynamics. For numerical implementation in the peridynamics framework, the non-ordinary state-based PD formulation of our space-time gauge theory is also presented briefly.

We organize the rest of the article as follows. In Section \ref{vd:stvd}, we present the space-time gauge theory and discuss the notion of gauge invariance under local translation and conformal transformation. We present the principle of virtual power and force-balance in Section \ref{vd:forceB}. Subsequently, in Section \ref{vd:constitutiveR}, the constitutive relations are derived. Section \ref{vd:PD} deals with non-ordinary state based PD formulation of our model. Illustrative numerical examples are provided in Section \ref{vd:numericalR} before concluding the article in Section \ref{vd:Conclusion}. Axisymmetric formulation for non-ordinary state-based peridynamics is presented briefly in Appendix \ref{axi_formulation}.

\section{Space-time gauge theory}
\label{vd:stvd}

This section presents the notion of gauge invariance under local translational and conformal (scaling) symmetries. For completeness, we briefly discuss continuum elasticity kinematics in the space-time setup in line with \cite{roy2020cauchy}, which includes a description of the reference and the current configurations as 4-dimensional differentiable manifolds, distance measures in both reference and current configurations, 4D deformation map, etc.  Finally, we non-trivially extend the space-time continuum elasticity formulation to the ductile damage problem exploiting local translational and conformal symmetries.

 \begin{figure}
 \centering
 \includegraphics[width=0.55\linewidth]{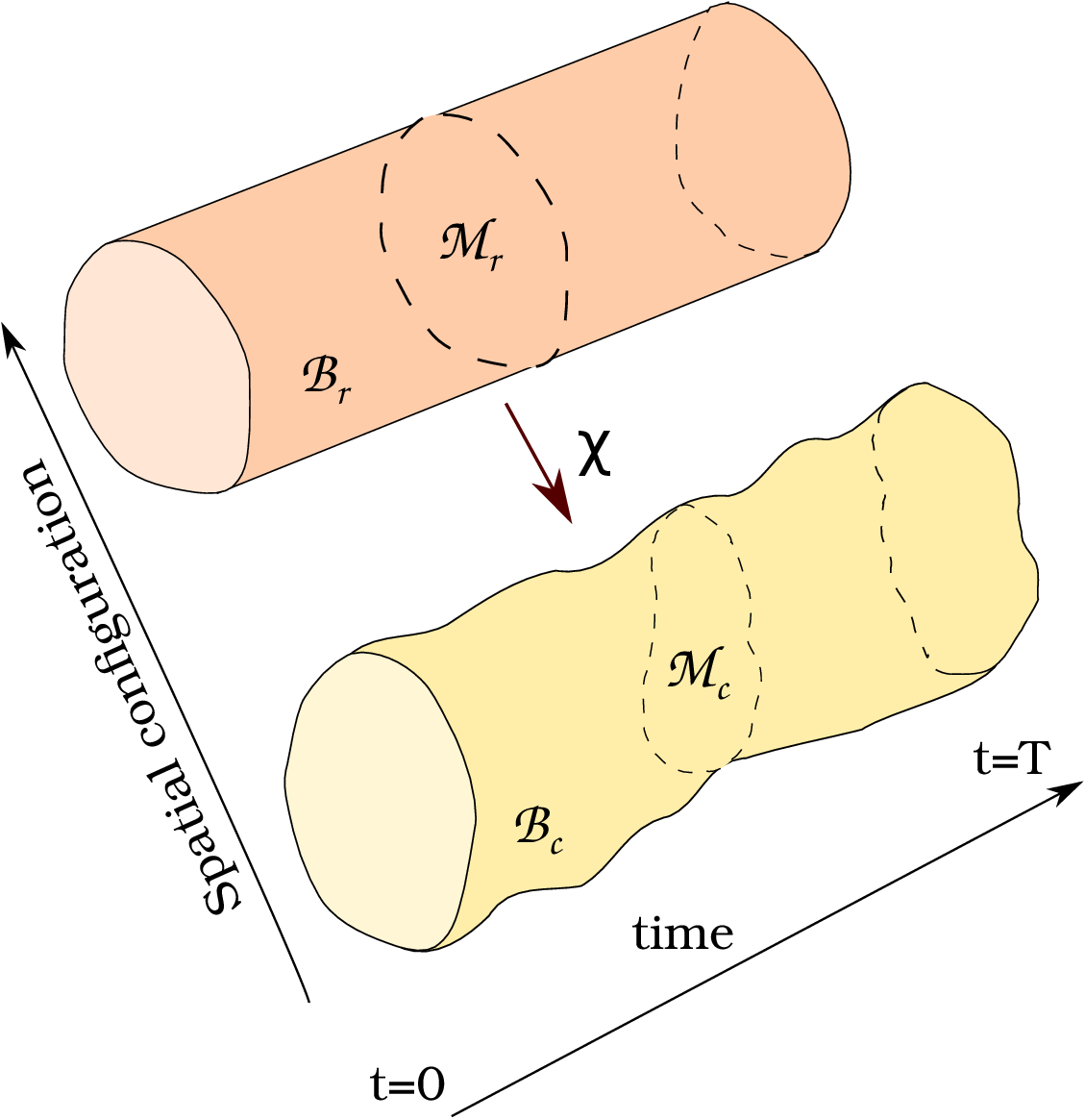}
 \caption{Schematic of reference and current configuration histories. Time coordinate is along the axis of the cylinder and cross-section at any time denotes the snapshot of spatial configuration.}
 \label{fig:configuration histories}
 \end{figure}

\subsection{Kinematics}

We consider the reference and the current configurations as 4-dimensional differentiable manifolds. The reference configuration history is a collection of self-similar copies of the reference configuration at every time $t$ where, $0 \leqslant t \in \mathbb{R}^+$, whereas the current configurations history can be described as a set that includes one reference configuration at $t=0$ and all the current configurations for time $t>0$. In space-time, time constitutes the first coordinate and the rest three are spatial coordinates.  Mathematically, the 4-dimensional reference configuration history is given as, $\mathfrak{B}_{r} = \left\lbrace \mathcal{M}_r|^t \right\rbrace_{t \in \mathbb{R}}$ and the current configurations history as, $\mathfrak{B}_{c} = \left\lbrace \mathcal{M}_c|^t \right\rbrace_{t \in \mathbb{R}}$, see  Fig. \ref{fig:configuration histories}.  $\mathcal{M}_r \subset \mathbb{R}^3$ and $\mathcal{M}_c \subset \mathbb{R}^3$ represent the reference and the  current configurations respectively at any time $t$. Time coordinate is along the axis of the cylindrical configuration as shown in Fig. \ref{fig:configuration histories} and cross sections through the reference and the current configuration histories at any time $t$, represent the reference and the current configurations of the body at that time. Deformation of time is not considered in the space-time formulation of classical continuum mechanics, i.e. time coordinate is the same in both reference and current configurations. $\bs{\tilde{G}}$ and $\bs{\tilde{g}}$ represent Minkowski metric tensors in $\mathfrak{B}_{r}$ and $\mathfrak{B}_{c}$ respectively.  They are given by, $\bs{\tilde{G}}=\bs{\tilde{g}}=\left[\begin{array}{cc} {1} & \bs{ 0}^{\mathsf{T} } \\ \bs{0} & \tb{-I} \end{array}\right]$ with $\bs{0}=\left(0\, \, 0\, \, 0\right)^{\mathsf{T}} $. $\tb{I}$ denotes the $(3\times3)$ identity matrix. Overhead tilde $\tilde{(\cdot)}$ is used to denote 4D objects, thus distinguishing them from 3D objects and coordinate indices on a 4D object vary from 1 to 4. The distance between two infinitesimally separated points in $\mathfrak{B}_{r}$ and $\mathfrak{B}_{c}$ are given as:

\begin{equation}
\left(dS\right)^{2} =c_{el}^{2} \, \left(dt\right)^{2} -\left(d\tilde{X}^{2} \right)^{2} -\left(d\tilde{X}^{3} \right)^{2} -\left(d\tilde{X}^{4} \right)^{2}
\end{equation}

\begin{equation}
\left(ds\right)^{2} =c_{el}^{2} \, \left(dt\right)^{2} -\left(d\tilde{x}^{2} \right)^{2} -\left(d\tilde{x}^{3} \right)^{2} -\left(d\tilde{x}^{4} \right)^{2} 
\end{equation} 
where, $\bs{\tilde{X}}$ and $\bs{\tilde{x}}$ are respectively the four dimensional coordinates of $\mathfrak{B}_{r}$ and $\mathfrak{B}_{c}$.  $\bs{X} = \left( \tilde{X}^2 , \tilde{X}^3 , \tilde{X}^4   \right)^{\mathrm{T}}$ and $\bs{x} = \left( \tilde{x}^2 , \tilde{x}^3 , \tilde{x}^4   \right)^{\mathrm{T}} $ represent the spatial coordinates of $\mathfrak{B}_{r}$ and $\mathfrak{B}_{c}$, respectively. We have,  $\tilde{x}^1 = \tilde{X}^1 = c_{el} t = \hat{t},$ where $\hat{t}$ has the unit of length. $c_{el}$ is a material parameter having the unit of velocity. $c_{el}$ ensures consistency among units of space and time coordinates. Next, we define a 4D deformation map that relates space-time reference and current configurations as follows: $\tilde{x}^{i} ={\tilde{\chi}} ^{i} \left(\bs{X},t\right)=\hat{\chi }^{i} \left(\bs{X},\hat{t}\right)$; where, $\bs{\tilde{\chi}}$ and $\bs{\hat{\chi}}$ respectively denote the deformation maps with real time and scaled time. Now using the 4-dimensional deformation map, the coordinate differentials in the current configuration history may be given as:

\begin{equation}
d\tilde{x}^1 = d \tilde{X}^1
\end{equation}

\begin{eqnarray} 
d\tilde{x}^{\alpha+1} =\frac{\partial \hat{\chi }^{\alpha+1} }{\partial \hat{t}} \: d\hat{t} + \frac{\partial \hat{\chi}^{\alpha+1} }{\partial X^{\beta} } \: dX^{\beta} =\hat{v}_{\alpha} \: d\hat{t} + {F}_{\alpha \beta} \: dX^{\beta} 
\label{eqD:coordinate diff. curr. config.-2} 
\end{eqnarray}
where, $ \alpha, \beta \in \left\lbrace1,2,3\right\rbrace $. $\bs{F}$ is the classical three dimensional deformation gradient. $\hat{v}_{\alpha} = \frac{\partial \hat{\chi }^{\alpha+1} }{\partial \hat{t}} = \frac{1}{c_{el}} v_{\alpha} $ is the scaled material velocity and $v_{\alpha}$ is the material velocity. Note that the temporal derivative of the deformation map in Eq. \ref{eqD:coordinate diff. curr. config.-2} does not appear in classical elasticity theory, wherein the deformed configuration is assumed as a snapshot of  $\mathfrak{B}_{c}$ at time $t$. Unlike our space-time formulation, time is not a coordinate in classical elasticity.   We are now ready to extend the space-time continuum formulation of elastodynamics to the case of ductile damage. 

\subsection{Gauge invariance under local translational symmetry} 

 The 4D deformation map $\bs{\tilde{\rchi}}$ is given as:

\begin{eqnarray}
\tilde{x}^1  = \tilde{\rchi}^1 \left( t \right)= \hat{\rchi}^1 \left( \hat{t} \right)\\
\tilde{x}^{\alpha+1}  = \tilde{\rchi}^{\alpha+1} \left(t, \bs{X} \right)= \hat{\rchi}^{\alpha+1} \left( \hat{t}, \bs{X} \right)
\end{eqnarray}
where, $\alpha \in \left\lbrace 1,2,3 \right\rbrace$. Once more, we recall that deformation in time is not considered in the space-time continuum elastodynamics, i.e. $\tilde{x}^1 = \tilde{X}^1$. Now  we consider the local translation in 4D $\bs{\tilde{\rchi}}$ map and attempt to construct the gauge invariant energy density under local translations in space and time. The local translation in the 4D deformation map may be written as:

\begin{equation}
^\backprime{\bs{\tilde{\rchi}}} = \bs{\tilde{\rchi}} \left(\hat{t}, \bs{X} \right) + \bs{\tilde{\Upsilon}} \left(\hat{t}, \bs{X} \right)  
\label{eqD: Translation in space-time}
\end{equation}

Denoting $ \tilde{\Upsilon}^1\left(\hat{t}, \bs{X} \right) = \tau \left( \hat{t},\bs{X} \right) $ and $ \tilde{\Upsilon}^{\alpha+1} \left(\hat{t}, \bs{X} \right) = {\Upsilon}^{\alpha} \left(\hat{t}, \bs{X} \right) $, we may write the local translation in the 4D deformation map more explicitly as:

\begin{eqnarray}
^\backprime{\hat{t}} = \hat{t} + \tau \left( \hat{t},\bs{X} \right)
\label{eqD:Translation in time}
\end{eqnarray}

\begin{equation}
^\backprime{\rchi}^{\alpha}= \rchi^{\alpha} \left(\hat{t}, \bs{X} \right) + {\Upsilon}^{\alpha} \left(\hat{t}, \bs{X} \right)  
\label{eqD: Translation in space}
\end{equation}
Note that in Eq. \ref{eqD:Translation in time}, the time coordinate $\hat{t}$ is the same at all material points.  $\tau$ and $\bs{\Upsilon}$ represent local (inhomogeneous) translations in time and space respectively. From Eqs. \ref{eqD:Translation in time} and \ref{eqD: Translation in space}, it is clear that, $d^\backprime{\hat{t}} \neq d \hat{t} $ and $d^\backprime{\rchi}^{\alpha} \neq d \rchi^{\alpha}$. We exploit minimal replacement to define the gauge covariant operator in place of ordinary partial derivative so as to restore the invariance of $d\hat{t}  $ and $d\rchi^\alpha$ under local translations in space-time.


\subsubsection{Minimal replacement}

 The gauge covariant definition of the ordinary time differential is given as:

\begin{eqnarray}
\nonumber
D\hat{t} & :=  & \mathcal{M} \big{\langle} d\hat{t} \big{\rangle}  \\
& = & d\hat{t} + \Psi_i d{\tilde{X}}^i
\end{eqnarray}
where the scalar-valued 1-form $\Psi$ may be written as, $\Psi = \Psi_t d\hat{t} +  \Psi_{\bs{X}} \bs{dX} $. By sacrificing some generality in the model, we may choose the space-time gauge covariant temporal derivative as:

\begin{equation}
D\hat{t} = c_{el} \left(1+ \Psi_t \right) dt
\label{eqD:CoVD_time}
\end{equation}
Specifically, the spatial component of the temporal gauge compensating field is assumed to be zero in Eq. \ref{eqD:CoVD_time}. A direct consequence of $\Psi_{\bs{X}}$ being dropped from the model is that local time translation is not directly coupled with local translations in space. The minimal replacement on the corresponding tangent space object of $d\hat{t}$ is given as:

\begin{equation}
\mathcal{M} \bigg{\langle} \frac{\partial}{\partial \hat{t}} \bigg{\rangle} = \frac{1}{\left(1+ \Psi_t \right)} \frac{\partial}{\partial \hat{t}}
\label{eqD: Minimal replacement on d/dt}
\end{equation}

Similarly the minimal replacement on $\partial_i{\rchi}^\alpha$ is given as:

\begin{equation}
\mathcal{M} \big{\langle} \partial_i{\rchi}^\alpha \big{\rangle} = \partial_i{\rchi}^\alpha - \bar{\Phi}^\alpha_i
\end{equation}

Minimal replacement introduces gauge compensating fields $\Psi$ and $\bs{\bar{\Phi}}$ in time and space, respectively; $i$ ranges from 1 to 4.  Note that, $D ^\backprime{\hat{t}} = D \hat{t} $ and $\mathcal{M} \big{\langle} d^\backprime{\rchi}^{\alpha} \big{\rangle}= \mathcal{M} \big{\langle} d \rchi^{\alpha} \big{\rangle}$ using the transformations given below in Eqs. \ref{eqD:transf. of Psi} and \ref{eqD:transf. of Phi},

\begin{equation}
^\backprime{\Psi} = \Psi - \partial_i {\tau}
\label{eqD:transf. of Psi}
\end{equation}
\begin{equation}
^\backprime{\bar{\Phi}}^\alpha_i= \bar{\Phi}^\alpha_i + \partial_i{\Upsilon}^\alpha
\label{eqD:transf. of Phi}
\end{equation}

\subsubsection{Gauge invariant kinematics}

We may write the 4D coordinate differentials in reference and current configurations using the space-time gauge covariant derivatives as follows:

\begin{equation}
d \bs{\tilde{X}} = \left(D \hat{t} \;\;\  d\tilde{X}^2 \;\;\  d\tilde{X}^3 \;\;\  d\tilde{X}^4  \right)
\end{equation}

\begin{equation}
d\bs{\tilde{x}} = \left(D \hat{t} \;\;\  d\tilde{x}^2 \;\;\  d\tilde{x}^3 \;\;\  d\tilde{x}^4  \right)
\end{equation}
Using the Minkowski metric, distance between two infinitesimally close points in the reference configuration, $\mathcal{M}_r$ and the current configuration, $\mathcal{M}_c$ may be computed by the expressions given below in Eqs. \ref{eqD: distance measure in reference configuration} and \ref{EqD: distance measure in current configuration}, respectively.

\begin{equation}
\left(dS \right)^2 = \tilde{G}_{ij} d\tilde{X}^i d\tilde{X}^j
\label{eqD: distance measure in reference configuration}
\end{equation}

\begin{equation}
\left(ds \right)^2 = \tilde{g}_{ij} d\tilde{x}^i d\tilde{x}^j
\label{EqD: distance measure in current configuration}
\end{equation}
We may express the coordinate differentials in the current configuration history in terms of those in the reference configuration as:

\begin{equation}
d\tilde{x}^1 = D \hat{t} = d\hat{t} + \Psi
\label{eqD: Current config. differentials in terms of ref. config-I}
\end{equation}

\begin{equation}
d\tilde{x}^{\alpha+1} = d x^\alpha = \left(  \frac{1}{\left(1+ \Psi_t \right)} \frac{\partial \tilde{\rchi}^{\alpha+1}}{ \partial \hat{t}} -\bar{{\Phi}}^\alpha_1 \right) D\hat{t} + \underbrace{\left( \frac{\partial \tilde{\rchi}^{\alpha+1}}{\partial X^\beta} -\Phi^\alpha_\beta \right)}_{\bs{F^e}} d X^\beta 
\label{eqD:Crefconfig-II}
\end{equation}
where, ${F}_{\alpha\beta} =  \frac{\partial \tilde{\rchi}^{\alpha+1}}{\partial X^\beta} $ denotes the classical deformation gradient and ${F}^e_{\alpha \beta} := \left( \frac{\partial \tilde{\rchi}^{\alpha+1}}{\partial X^\beta} -\Phi^\alpha_\beta \right)$. The scaled material velocity is given by, $\bs{\hat{v}}_{\alpha} =  \frac{\partial \tilde{\chi}^{\alpha+1}}{ \partial \hat{t}}$ and $\Phi^\alpha_\beta = \bar{\Phi}^\alpha_{1+\beta}. $ Note that, $ \alpha, \beta \in \left\lbrace 1,2,3 \right\rbrace$. We may express Eqs. \ref{eqD: Current config. differentials in terms of ref. config-I} and \ref{eqD:Crefconfig-II} in a more compact form as:

{\renewcommand{\arraystretch}{1.7}
\begin{equation}
\begin{bmatrix}
D \hat{t} \, \\
d\bs{x} 
\end{bmatrix}
=
\begin{bmatrix}
\left(1+ \Psi_t \right) & \bs{0}^{\mathsf{T}} \\
\bs{\hat{v}} - \left(1+ \Psi_t \right) \bar{{\Phi}}_1 & \bs{F}^e
\end{bmatrix}
\begin{bmatrix}
d\hat{t} \\
d\bs{X} 
\end{bmatrix}
\label{eqD:Current config. diff. in terms of ref. config}
\end{equation}}
\quad 

From Eq. \ref{eqD:Current config. diff. in terms of ref. config}, we may notice that the temporal components of the gauge compensating fields $\bar{{\Phi}}_1$ and $\Psi_t$ modify velocity which may lead to a correction in the pseudo-forces and momenta; see \cite{roy2020cauchy}. In line with our presently adopted strategy of space-time geometric decoupling -- assuming $\bar{{\Phi}}^\alpha_1 = 0$, we may obtain the 4D pulled back metric using Eqs. \ref{eqD:Current config. diff. in terms of ref. config} in Eq. \ref{EqD: distance measure in current configuration} as:

{\renewcommand{\arraystretch}{2.5}
\begin{equation}
\bs{\tilde{C}} = \begin{bmatrix}
\left(1+\Psi_t \right)^2 - \bs{\hat{v}}^{\mathsf{T}} \bs{\hat{v}} &  -  \bs{\hat{v}}^{\mathsf{T}} \bs{F}^e \\
- \bs{F}^{e\mathsf{T}} \bs{\hat{v}}   &  - \bs{C}^e
\end{bmatrix}_{\left(4 \times 4 \right)}
\end{equation}}
where, $\bs{C}^e = \bs{F}^{e\mathsf{T}} \bs{F}^e $. Now defining the gauge transformation, $\bs{\Phi} = \bs{F} \bs{\hat{\Phi}}$ and using the second term on the right hand side of Eq. \ref{eqD:Crefconfig-II}, We get:

\begin{equation}
\bs{F}^e = \bs{F} \left(\tb{I} - \bs{\hat{\Phi}} \right)
\label{eqD: Reln between Fe, F and Phi}
\end{equation}

Denoting $\left(\tb{I} - \bs{\hat{\Phi}} \right) = \bs{F}^{p-1}$ in Eq. \ref{eqD: Reln between Fe, F and Phi}, the multiplicative decomposition of the deformation gradient $\bs{F}= \bs{F}^e \bs{F}^p$, is recovered, even though an introduction of the so called intermediate configuration (the structure space, which is not a manifold) is not required explicitly in the present model. Next, from the existing gradient based plasticity theories \citep{lele2009large,gurtin2010mechanics}, we may write $\bs{F}^e$ and $\bs{F}^{p}$ in terms of the equivalent plastic strain $(\gamma^p)$ as follows:

\begin{equation}
\bs{F}^e = \bs{F} \, \exp \Big( - \gamma^p \bs{N}^p  \Big)
\label{eqD:Fe in terms of gammaP}
\end{equation}

and,

\begin{equation}
\bs{F}^{p}=  \exp \Big(\gamma^p \bs{N}^p \Big)
\label{eqD: Fp and gamma_p relation}
\end{equation}

The equivalent plastic strain $\gamma^p$ is defined as:  

\begin{equation}
\gamma^p = \int_t \nu^p \, dt
\label{eqD:equivalent plastic_strain}
\end{equation}

where, $\bs{N}^p$ and $\nu^p$ denote plastic flow direction tensor and plastic strain rate, respectively. Using Eq. \ref{eqD: Fp and gamma_p relation}, we may re-write the gauge compensating field as:

\begin{equation}
\hat{\bs{\Phi}} =  \tb{I} - \exp \Big( -\gamma^p \bs{N}^p \Big) 
\label{eqD: Reln between Phi and Gamma_p}
\end{equation}

The motivation for expressing the gauge compensating field $\hat{\bs{\Phi}}$, which has emerged due to local translational symmetry in terms of equivalent plastic strain is two-fold. First one is to reconcile the present theory with the existing viscoplasticity theories which is useful for experimental verification and calibration purposes. Second and the important one is to give a geometric interpretation of the plastic quantifier. Now we construct the invariant energy density associated with the energetic part of the plastic deformation using the gauge compensating field ($\hat{\bs{\Phi}}$) in the next section.

\subsubsection{Invariant energy density}

This subsection presents the construction procedure of the invariant energy density using various components of the 4D pulled back metric $\bs{\tilde{C}}$ and the invariants of $\bs{C}^e$ under space-time local translation. The invariant energy density may be written as:

\begin{equation}
\bar{\Pi} =  \Pi^e_k +  \Pi^e  + \Pi^z  
\label{eqD:Total invariant energy,Pi}
\end{equation}
The first term in the energy density may be constructed using the (1,1) component of the 4D pulled back metric:

\begin{equation}
\Pi^e_k  = k_{11} \left(1+\Psi_t \right)^2 -\frac{k_{12}}{c_{el}^2} \,\bs{v}^{\mathrm{T}} \bs{v}
\end{equation}
The second term in Eq. \ref{eqD:Total invariant energy,Pi} representing the energy density due to the elastic deformation can be decomposed as the sum of the volumetric and the isochoric parts:

\begin{equation}
\Pi^e =  \Pi^e_{vol} + \Pi^e_{ic} 
\label{eqD:energy_density_elas} 
\end{equation}
where the  volumetric part may be expressed as: 

\begin{equation}
\Pi^e_{vol} = k_2 \left[\left(\text{det} \bs{C}^e\right)^{\frac{1}{2}} -1 \right]^2
\end{equation}
The isochoric part is given as:

\begin{equation}
\Pi^e_{ic} = k_3 \left[C^e_{\mu \nu} G^{\mu \rho} G^{\nu \sigma} C^e_{\rho \sigma}\left(\text{det} \bs{C}^e \right)^{-\frac{2}{3}} - 3 \right] 
\end{equation}

In the present theory, we take the material parameters the same as in \cite{roy2020cauchy}, i.e. $k_{12} = 0.5 \rho c_{el}^2,\, k_2 = \frac{1}{2}(\lambda + 2\mu/n_s)$ and $k_3 = \mu/8$.  $n_s$ denotes the dimension of the space. The Lam\'e parameters are $\lambda$ and $\mu$. The material constants $k_{11}$ is taken as zero. Finally, we construct the invariant energy density associated with the gauge compensating field $\hat{\bs{\Phi}}$. From Eq. \ref{eqD: Reln between Phi and Gamma_p}, we may write:

\begin{equation}
\bs{d}\hat{\bs{\Phi}} =  \exp\Big(-\gamma^p \bs{N}^p  \Big) \left[\bs{N}^p \left(d\gamma^p + \frac{1}{c_{el} \left(1+\Psi_t \right)} \dot{\gamma}^p \mathrm{D}\hat{t} \right) \right]
\end{equation}

 Consider a 1-form, $\bs{z} = z_{i} {d X}^{i}$,\: where $i$ varies from 1 to 4. The 1-form, $\bs{z}$ is  defined as:

\begin{equation}
\bs{z} :=  \left( \bs{F}^{p-1} \bs{N}^{p} \right)^{-1} \bs{d} \bs{\Phi}
\end{equation}
Then the invariant energy density, $\Pi_z$ is given as,
\begin{eqnarray}
\nonumber
\Pi^{z} &=&  k_z \, {z}_\mu \tilde{G}^{\mu \nu} {z}_\nu\\
&=&   \frac{k_z}{c_{el}^2 \left(1+ \Psi_t \right)^2} \left(\dot{\gamma}^p \right)^2  -  k_z \bs{\nabla} \gamma^p \cdot \bs{\nabla}\gamma^p
\label{eqD:Pi^z}
\end{eqnarray}

\subsection{Conformal transformation and gauge invariance}

In this subsection, we discuss the conformal symmetry along with scale-invariant transformation on the reference metric. First, we present the concept of local symmetry under conformal transformation and outline the minimal replacement procedure, which is followed by a Weyl-like condition and transformations of the gauge compensating field variables. Finally, the construction of a gauge invariant curvature tensor and invariant energy density pertaining to the ductile deformation is presented.

\subsubsection{Scaling symmetry and Weyl-like condition}

By global conformal transformation on the  metric, we imply a uniform scaling of the reference metric, i.e. $^\backprime{\bs{\tilde{G}}} = e^{f} \bs{\tilde{G}}$, where the scaling variable $f$ is a scalar  constant. This implies that the metric $\bs{\tilde{G}}$  is scaled homogeneously over the entire body and under such global/homogeneous conformal transformation, the Christoffel symbols associated with the metric compatible connection will remain invariant. However, under local conformal transformation of the reference metric involving a non-uniform scaling, $ ^\backprime{\bs{\tilde{G}}} = e^{f(\bs{X},t)} \bs{\tilde{G}}$, where $f$ is a smooth scalar-valued function in space-time, the  invariance of the metric-compatible connection is lost and some additional terms arise. Minimal replacement is required to compensate for the additional terms. We define the minimal replacement by imposing the Weyl-like condition on the reference configuration metric as follows:

\begin{equation}
\mathcal{M} \big\langle \partial_\sigma {\tilde{G}}_{\mu \nu} \big\rangle = \partial_\sigma {\tilde{G}}_{\mu \nu} +  2\Xi_\sigma \tilde{G}_{\mu \nu}
\label{minimal replacement}
\end{equation}
where, $\bs{\Xi}=\Xi_\mu \text{d}X^{\mu}$ denotes a 1-form field in space-time. Using the minimal replacement constructed in Eq. \ref{minimal replacement}, the modified Christoffel symbols are given as:

\begin{equation}
\overline{\Gamma}^\rho _{\sigma \mu} = \Gamma^\rho _{\sigma \mu} + \tilde{G}^{\rho \nu} \left[\tilde{G}_{\mu \nu}  {\Xi}_\sigma + \tilde{G}_{\sigma \nu} \Xi_\mu  - \tilde{G}_{\sigma \mu} \Xi_\nu \right]
\label{modified connection}
\end{equation}
Note that the indices, $(\rho, \sigma, \mu, \nu)$ vary from 1 to 4. With this, the modified connection vis-\`a-vis the Christoffel symbol in Eqn. \ref{modified connection} remains unchanged i.e. 
$^\backprime\overline{\Gamma}^\rho _{\sigma \mu} = \overline{\Gamma}^\rho _{\sigma \mu}$, under the following simultaneous transformations on $\bs{\tilde{G}}$ and $\bs{\Xi}$,

\begin{equation}
\bs{^\backprime{{\tilde{G}}}} = e^{f(\bs{X},t)} \bs{\tilde{G}}
\label{transformation on G}
\end{equation}

\begin{equation}
\bs{^\backprime{\Xi}} = \bs{\Xi} - \frac{1}{2} df
\label{transformation on gamma}
\end{equation}

 The 1-form, $\bs{\Xi}$ may be decomposed as the sum of an anti-exact part and an exact differential of some scalar valued function, $\bs{\Xi} = \mathfrak{F}(\phi)\, \bs{\Xi}_{ae} + d\phi$; where, $\bs{\Xi}_{ae}$ is the anti-exact part of the 1-form and $d\phi$ represents the exact differential  of the real valued scalar function $\phi$ in space-time. $\mathfrak{F}(\phi)$ represents a scalar valued function of $\phi$. In component form, $\left(\Xi_{ae}\right)_\mu = \left(\bar{\Xi}_t \:\: \bar{\Xi}_1 \:\: \bar{\Xi}_2 \:\: \bar{\Xi}_3\right)$; wherein the temporal component is given as, $\left(\Xi_{ae}\right)_1 = \bar{\Xi}_t d\hat{t}$ and $\bs{\bar{\Xi}}= \bar{\Xi}_\alpha dX^\alpha$ denotes the spatial components, $\alpha \in \{1,2,3\}$. If the anti-exact part of the 1-form, $\bs{\Xi}_{ae}$  is zero, then $\bs{\Xi} = d\phi$ and the transformation of $\phi$ is given as: $^\backprime{\phi} = \phi -\frac{1}{2} f$.

\subsubsection{Gauge invariant curvature tensor}
Using the modified connection defined in Eq. \ref{modified connection}, we construct a space-time gauge invariant curvature given as:

\begin{equation}
\bar{R}^\rho _{\sigma \mu \nu} =  \partial_\mu \overline{\Gamma}^\rho _{\nu \sigma} -\partial_\nu \overline{\Gamma}^\rho _{\mu \sigma} + \overline{\Gamma}^\rho _{\mu \lambda} \overline{\Gamma}^\lambda _{\nu \sigma} -\overline{\Gamma}^\rho _{\nu \lambda} \overline{\Gamma}^\lambda _{\mu \sigma}
\label{eq:space-time curvature}
\end{equation}
Contracting the $1^{st}$ and $3^{rd}$ indices in the curvature given in Eq. \ref{eq:space-time curvature}, we get $\bar{R}^\rho _{\sigma \rho \nu} = \bar{R} _{\sigma \nu}$, which we may express explicitly in terms of the Christoffel symbols as follows:

\begin{equation}
\bar{R} _{\sigma \nu} = \partial_\rho \overline{\Gamma}^\rho _{\nu \sigma} -\partial_\nu \overline{\Gamma}^\rho _{\rho \sigma} + \overline{\Gamma}^\rho _{\rho \lambda} \overline{\Gamma}^\lambda _{\nu \sigma} -\overline{\Gamma}^\rho _{\nu \lambda} \overline{\Gamma}^\lambda _{\rho \sigma}
\end{equation}
We construct scalar gauge invariant curvature by contracting the second order gauge invariant curvature tensor, $\bar{R} _{\sigma \nu}$ with metric as follows: $\bar{R} = \tilde{G}^{\sigma \nu} \bar{R}_{\sigma \nu}$. Since we have a referential Minkowski metric, the Riemmanian curvature associated with the metric compatible connection is zero; but the gauge invariant curvature is not. The non-zero nature of the gauge invariant scalar curvature is solely due to the minimal replacement used to make the connection invariant. An explicit expression for the curvature in terms of the scalar valued function $\phi$ and the anti-exact part $\bs{\Xi}_{ae}$ is given as, $\bar{R} = \bar{R}_t + \bar{R}_X$ as follows:

\begin{equation}
\bar{R}_X = 6 \left[ \mathfrak{F}(\phi) \nabla \cdot \bs{\bar{\Xi}} + {\nabla} \mathfrak{F}(\phi) \cdot \bs{\bar{\Xi}} + \nabla \cdot \nabla \phi + \mathfrak{F}(\phi)^2 \bs{\bar{\Xi}} \cdot \bs{\bar{\Xi}} + 2\mathfrak{F}(\phi)\bs{\bar{\Xi}} \cdot \nabla \phi + \nabla\phi \cdot \nabla\phi  \right]
\end{equation}

\begin{equation}
\bar{R}_t = - \frac{6}{c_{el}} \left[\mathfrak{F}(\phi) \frac{\partial \Xi_t}{\partial t} +  \frac{1}{c_{el}} \left( \ddot{\phi} +  \dot{\phi}^2 \right)  +  \left( \frac{\partial \mathfrak{F}(\phi)}{\partial t} + \mathfrak{F}(\phi) \dot{\phi} \right) \Xi_t\right] +  6\mathfrak{F}(\phi)^2 \Xi_t^2    
\end{equation}

It is worth mentioning that, $\bs{\Xi}_{ae}$ is related with the polarization vector and scalar potential of the electric field; see \cite{roy2020cauchy}. Our present focus is however on a basic model for ductile deformation and thus we do not consider the electromagnetic effects of dielectrics. Indeed a unified electro-magneto-mechanical ductile damage model is kept as part of the future scope of work. We consider $\bs{\Xi}_{ae}$ as constant in both space and time, i.e. $\bs{\bar{\Xi}} = \bs{\hat{e}}_\alpha, \alpha\in \{1,2,3\}$ where $\bs{\hat{e}}$ represents a unit vector in space and $\Xi_t = k_p $. $k_p$ is a real scalar constant. Assuming $\mathfrak{F}(\phi)=k_{\phi} \left(1- \phi \right)$ as a linear function of $\phi$ for simplification and neglecting the higher order derivative terms, we may recast $\bar{R}_X$ and $\bar{R}_t$  as:

\begin{equation}
\bar{R}_X = 6 \left[k_\phi  \bs{\hat{e}} \cdot \nabla\phi \left(1 - 2 \phi\right)  +  k_{\phi}^2 \left(1- \phi \right)^2  + \nabla\phi \cdot \nabla\phi  \right]
\label{eq:R_X}
\end{equation}

\begin{equation}
\bar{R}_t =  \frac{6}{c_{el}} \left[- \frac{1}{c_{el}}  \dot{\phi}^2 + k_p k_\phi \phi\, \dot{\phi} \right] - 6k_p^2 k_{\phi}^2 \left(1- \phi \right)^2 
\label{eq:R_t}
\end{equation}

The invariant energy density associated with the gauge invariant curvature may be written as:

\begin{equation}
\Pi^{\bar{R}} = k_{x} \bar{R}_X + k_{t} \bar{R}_t
\end{equation}

\section{Principle of virtual-power and force balance} 
\label{vd:forceB}
 
Using the principle of virtual power, we obtain the macro- and micro-force balances pertaining to the ductile damage response in a finite deformation setup. The principle of virtual power postulates a balance between the external power $\mathcal{W}^{\text{ext}}(\mathcal{P}_t)$ applied on $\mathcal{P}_t$ and the internal power $\mathcal{W}^{\text{int}}(\mathcal{P}_t)$ expended within it, where $\mathcal{P}_t$ is an arbitrary subregion  of the reference body  at time $t$. The external and internal power are given as:

\begin{equation}
\mathcal{W}^{\text{ext}} \left(\mathcal{P}_t \right) = \int_{\partial \mathcal{P}_t} \hat{\tb{t}}(\bs{n}) \cdot \dot{\bs{\rchi}} \, dA + \int_{\mathcal{P}_t} \left[ \tb{b} \cdot \dot{\bs{\rchi}} - \, \zeta_{\gamma} \ddot{\gamma}^p  \dot{\gamma}^p - \zeta_\phi \ddot{\phi}\,\dot{\phi}\right] dV   +  \int_{\partial \mathcal{P}_t} \left[\Lambda_\gamma (\bs{n}) \, \dot{\gamma}^p + \Lambda_\phi(\bs{n})\, \dot{\phi} \right] dA 
\end{equation}

\begin{equation}
\mathcal{W}^{\text{int}} \left(\mathcal{P}_t \right) = \int_{\mathcal{P}_t}  \left[\tb{T}^e \colon \dot{\bs{F}}^e + \pi_\gamma \dot{\gamma}^p + \pi_\phi \dot{\phi} + \bs{\xi}_\gamma \cdot \bs{\nabla} \dot{\gamma}^p + \bs{\xi}_\phi \cdot \bs{\nabla} \dot{\phi} \right] dV  
\label{eqD:int_power}
\end{equation}
Note that $\tb{b}$ represents the macroscopic body force (including inertia) given as, 

\begin{equation}
\tb{b} = \tb{b}_0 - \rho \ddot{\bs{\rchi}} 
\label{eqD:macroscopic body force}
\end{equation}
where, $\tb{b}_0$ denotes the conventional body force. $\zeta_\gamma$ and $\zeta_\phi$ are the coefficients of micro or defect inertia  due to the motion of defects such as dislocations and micro-cracks, respectively. $\Lambda_\gamma$ and $\Lambda_\phi$ are the external microscopic tractions on $\partial{\mathcal{P}_t}$, power conjugate with $\dot{\gamma}^p$ and $\dot{\phi}$ respectively. The scalar microstress $\pi_\gamma$  and the vector microstress $\bs{\xi}_\gamma$ are power conjugates with  $\dot{\gamma}^p$ and $\bs{\nabla} \dot{\gamma}^p $, respectively. Similarly the scalar microstress $\pi_\phi$  and the vector microstress $\bs{\xi}_\phi$ are power conjugates with  $\dot{\phi}$ and $\bs{\nabla} \dot{\phi} $, respectively.  Denoting the generalised virtual velocity by $\mathcal{V}(\bar{\bs{\rchi}}, \bar{\bs{F}}^e,\bar{\gamma}^p, \bar{\phi})$. $\bar{\bs{\rchi}}, \bar{\bs{F}}^e,\bar{\gamma}^p$ and $\bar{\phi}$ are the virtual velocities corresponding to $\dot{\bs{\rchi}}, \dot{\bs{F}}^e,\dot{\gamma}^p$ and $\dot{\phi}$, respectively. From the principle of virtual power, we may write:

\begin{equation}
\mathcal{W}^{\text{ext}}(\mathcal{P}_t,\mathcal{V}) = \mathcal{W}^{\text{int}}(\mathcal{P}_t,\mathcal{V}); \;\;  \forall \text{  generalized virtual velocities} \:  \mathcal{V}.
\label{eqD:virtual power balance-2}
\end{equation}


To derive the macro-force balance, consider $\bar{\gamma}^p = 0 $ and $\bar{\phi}=0$. Then from Eq. \ref{eqD:virtual power balance-2}, we may write:

\begin{equation}
\int_{\partial \mathcal{P}_t} \hat{\tb{t}}(\bs{n}) \cdot \bar{\bs{\rchi}}  dA + \int_{\mathcal{P}_t}  \tb{b} \cdot \bar{\bs{\rchi}} dV  = \int_{\mathcal{P}_t}  \tb{T}^e \colon \bar{\bs{F}}^e  dV 
\label{eqD:macro_force_balance}
\end{equation}
Using $\bar{\bs{F}}^e = \left(\bs{\nabla} \bar{\bs{\rchi}}\right) \bs{F}^{p-1}$, we may rewrite Eq. \ref{eqD:macro_force_balance} as:

\begin{equation}
\int_{\partial \mathcal{P}_t} \hat{\tb{t}}(\bs{n}) \cdot \bar{\bs{\rchi}}  dA + \int_{\mathcal{P}_t}  \tb{b} \cdot \bar{\bs{\rchi}} dV  = \int_{\mathcal{P}_t}  \tb{T}^e \bs{F}^{p-\mathsf{T}} \colon  \bs{\nabla} \bar{\bs{\rchi}}\, dV 
\label{eqD:macro_forceB}
\end{equation}  
Denoting, $\tb{T} = \tb{T}^e \bs{F}^{p-\mathsf{T}}$ and applying divergence theorem in Eq. \ref{eqD:macro_forceB}, we may write:

\begin{equation}
 \int_{\partial \mathcal{P}_t} \left[ \hat{\tb{t}}(\bs{n})- \tb{T} \bs{n}\right] \cdot \bar{\bs{\rchi}}  dA + \int_{\mathcal{P}_t} \left( \bs{\nabla} \cdot \tb{T}  + \tb{b}  \right) \cdot \bar{\bs{\rchi}} dV = \bs{0}
 \label{eqD:macroforce balance intVariable-1}
\end{equation}
Upon localization of Eq. \ref{eqD:macroforce balance intVariable-1} and using Eq. \ref{eqD:macroscopic body force}, we get:

\begin{equation}
 \bs{\nabla} \cdot \tb{T}  + \tb{b}_0 = \rho \ddot{\bs{\rchi}}
 \label{eqD:macroforce balance intVariable}
\end{equation}
The macroscopic boundary condition may be given as:
   \begin{equation}
   \hat{\tb{t}} (\bs{n}) = \tb{T} \bs{n} \;\;\;\; \text{on} \;\; {\partial \mathcal{P}_t} 
   \end{equation}

For deriving the microscopic force balance associated with plasticity, consider $\bar{\bs{\rchi}} = 0$ and $\bar{\phi}=0$. Then from Eq. \ref{eqD:Fe in terms of gammaP}, we may write:

\begin{equation}
\bar{\bs{F}}^e = - \bar{\gamma}^p \bs{F}^e \bs{N}^p
\end{equation}
Thus, Eq. \ref{eqD:virtual power balance-2}, takes the following form:

\begin{eqnarray}
 \int_{\partial \mathcal{P}_t}  \Lambda_\gamma (\bs{n})  \bar{\gamma}^p \, dA  - \int_{\mathcal{P}_t}  \zeta_\gamma \, \ddot{\gamma}^p \, \bar{\gamma}^p \, dV  = \int_{\mathcal{P}_t}  \left[\pi_\gamma  -  \bs{F}^{e\mathsf{T}} \tb{T}^e  \colon \bs{N}^p \right] \bar{\gamma}^p  dV
 +\int_{\mathcal{P}_t} \bs{\xi}_\gamma  \cdot \bs{\nabla} \bar{\gamma}^p \, dV
\end{eqnarray}
Applying divergence theorem, we get:

\begin{eqnarray}
\int_{\mathcal{P}_t} \left[\bs{\nabla} \cdot \bs{\xi}_\gamma -  \zeta_\gamma \, \ddot{\gamma}^p \right] \bar{\gamma}^p \, dV +  \int_{\partial \mathcal{P}_t}  \left[ \Lambda_\gamma (\bs{n}) - \bs{\xi}_\gamma  \cdot \bs{n} \right] \bar{\gamma}^p \, dA 
  +  \int_{\mathcal{P}_t}  \left[ \bs{F}^{e\mathsf{T}} \tb{T}^e  \colon \bs{N}^p - \pi_\gamma  \right] \bar{\gamma}^p dV = 0
 \label{eqD:microforce balance intVariable-1}
\end{eqnarray}

Employing localization of Eq. \ref{eqD:microforce balance intVariable-1} and using $\tb{T} = \tb{T}^e \bs{F}^{p-\mathsf{T}}$, we get the microscopic force balance for $\gamma^p$ as:

\begin{equation}
\pi_\gamma  +  \zeta_\gamma \, \ddot{\gamma}^p = \bs{\nabla} \cdot \bs{{\xi}}_\gamma + \left( \bs{F}^{e\mathsf{T}} \tb{T} \bs{F}^{p\mathsf{T}} \right)   \colon \bs{N}^p
 \label{eqD:microforce balance intVariable}
\end{equation}

and,

\begin{equation}
 \Lambda_\gamma (\bs{n}) = \bs{{\xi}}_\gamma \cdot \bs{n} \;\;\;\; \text{on} \;\; {\partial \mathcal{P}_t} 
\end{equation}

Since $\bs{N}^p$ is a symmetric and deviatoric tensor \citep{lele2009large,gurtin2010mechanics}, we see that on tensor contraction with $\bs{N}^p $, only the symmetric and deviatoric part of $ \bs{F}^{e\mathsf{T}} \tb{T} \bs{F}^{p\mathsf{T}}$ is nonzero. The symmetric part is defined as, $ \tb{T}^p := \;\; \text{sym}(\bs{F}^{e\mathsf{T}} \tb{T} \bs{F}^{p\mathsf{T}} )$. Next, considering $\bar{\bs{\rchi}} = 0$ and  $\bar{\gamma}^p = 0 $, we may derive the microscopic force balance pertaining to damage. From Eq. \ref{eqD:virtual power balance-2}, we may write:

\begin{eqnarray}
\int_{\partial \mathcal{P}_t} \Lambda_\phi(\bs{n})\, \bar{\phi} dA - \int_{\mathcal{P}_t} \zeta_\phi \ddot{\phi}\,\bar{\phi} dV = \int_{\mathcal{P}_t}  \bs{\xi}_\phi \cdot \bs{\nabla} \bar{\phi}  dV 
+ \int_{\mathcal{P}_t} \pi_\phi   \bar{\phi} dV
\label{eqD:microscopic force balance_phi-1}
\end{eqnarray}
Applying divergence theorem to the first term on the right hand side of Eq. \ref{eqD:microscopic force balance_phi-1}, we get:

\begin{eqnarray}
\int_{\partial \mathcal{P}_t} \left[\Lambda_\phi(\bs{n})- \bs{\xi}_\phi \cdot\bs{n} \right] \bar{\phi}\,dA -\int_{\mathcal{P}_t} \zeta_\phi \ddot{\phi} \bar{\phi} \,dV  =
 -\int_{\mathcal{P}_t} \left( \bs{\nabla} \cdot \bs{\xi}_\phi \right) \bar{\phi} \,dV +
\int_{\mathcal{P}_t} \pi_\phi \bar{\phi} \, dV
\label{eq:microscopic force balance_phi-2}
\end{eqnarray}

Upon localization of Eq. \ref{eq:microscopic force balance_phi-2}, we get the microscopic force balance for $\phi$ as:

\begin{equation}
\zeta_\phi \, \ddot{\phi} =  \bs{\nabla} \cdot \bs{{\xi}}_\phi - \pi_\phi
\label{eq:microscopic force balance_phi}
\end{equation}

and,

\begin{equation}
\Lambda_\phi(\bs{n}) = \bs{{\xi}}_\phi \cdot\bs{n} \;\;\;\; \text{on} \;\; {\partial \mathcal{P}_t} 
\end{equation}

\section{Constitutive relations}
\label{vd:constitutiveR}

In this section, we derive the constitutive relations by applying the two laws of thermodynamics. Constitutive closure to the equations of motion is typically achieved by describing a stress-strain affinity for the given material. The goal of the constitutive closure is to express $\tb{T}, \bs{\xi}_\gamma, \bs{\xi}_\phi, \pi_\gamma$ and $\pi_\phi$ in terms of the geometric quantities and derive closed-form expressions for $\zeta_\gamma$ and $\zeta_\phi$. 

\subsection{Evolution of internal energy density}
We derive the evolution of the internal energy density using the first law of thermodynamics, i.e., considering the conservation of total energy of the thermodynamic system and its surroundings \citep{tadmor2012continuum}. Let $\mathcal{E}$ and $W$ denote the total energy and the work done by the external forces on the thermodynamic system respectively. $Q$ is the heat supplied to the system. Then employing the first law of thermodynamics, we may write as: $d \mathcal{E}  = \dbar W + \dbar Q $; where, the total energy $\mathcal{E}$ equals the sum of the kinetic energy $\mathcal{K}$ and the internal energy $U$. Note that $W$ and $Q$ are not state functions, i.e. they depend on the thermodynamic path. The ineaxct differential $\dbar(\cdot)$ represents path dependence. The internal energy may be written in terms of the specific internal energy $e$ as, $U = \int_{\mathcal{P}_t} \rho e dV$. Considering the rate of heat supply $\mathcal{R} := \dbar Q/dt$ and the rate of external work ($ \dbar W /dt$), we may recast the first law of thermodynamics in rate form as follows:

\begin{equation}
\frac{\text{d} \mathcal{E}}{\text{d}t} = \frac{\dbar W}{dt} + \mathcal{R}
\label{eqD:First law in rate form}
\end{equation}
The rates of change of the kinetic and the internal energy may be expressed as:
\begin{equation}
\frac{d\mathcal{K}}{dt} = \int_{\mathcal{P}_t}  \rho \bs{v} \cdot \bs{\dot{v}}  dV
\label{eqD: dK/dt}
\end{equation}

\begin{equation}
  \frac{dU}{dt} = \int_{\mathcal{P}_t} \rho \dot{e} dV
  \label{eqD:du/dt}
\end{equation}

and,
\begin{equation}
\frac{\dbar W}{dt} = \int_{\partial \mathcal{P}_t} \bs{Tn} \cdot \bs{v} dA + \int_{\mathcal{P}_t} \tb{b}_0 \cdot \bs{v} dV + \int_{\partial\mathcal{P}_t} \left[ \Lambda_\gamma(\tb{n}) \, \dot{\gamma}^p + \Lambda_\phi(\tb{n}) \dot{\phi}\right] dA
\label{eqD:P^ext}
\end{equation}

Using $ \Lambda_\gamma(\tb{n}) = \bs{\xi}_\gamma  \cdot \tb{n}$ and  $\Lambda_\phi(\tb{n}) = \bs{\xi}_\phi  \cdot \tb{n} $ in Eq. \ref{eqD:P^ext} and applying divergence theorem to the boundary terms, we get:

\begin{eqnarray}
\nonumber
\frac{\dbar W}{dt} = \int_{\mathcal{P}_t} \left[ \left(\bs{\nabla}_X \cdot \tb{T} + \tb{b}_0 \right) \cdot \bs{v} + \tb{T}\, \bs{\colon} \bs{\dot{F}}  +  \left(\bs{\nabla}_X \cdot \bs{\xi}_\gamma \right) \dot{\gamma}^p + \bs{\xi}_\gamma \cdot  \bs{\nabla}_X \dot{\gamma}^p \right] dV \\
+ \int_{\mathcal{P}_t} \left[ \left(\bs{\nabla}_X \cdot \bs{\xi}_\phi \right) \dot{\phi} + \bs{\xi}_\phi \cdot \bs{\nabla}_X \dot{\phi} \right] dV
\label{eqD:external work}
\end{eqnarray}
Considering the heat flux vector $\bs{q}$ and heat source $h$, the thermal power  may be given as:

\begin{equation}
\mathcal{R} = \int_{\mathcal{P}_t} \left( \rho h - \bs{\nabla}_X \cdot \bs{q} \right) dV
\label{eqD:thermal power}
\end{equation}
Using  Eqs. \ref{eqD: dK/dt}-\ref{eqD:thermal power} in \ref{eqD:First law in rate form}, we may write:

\begin{eqnarray}
\nonumber
\int_{\mathcal{P}_t} \rho \dot{e} dV +  \int_{\mathcal{P}_t}  \rho \bs{v} \cdot \bs{\dot{v}}  dV =  \int_{\mathcal{P}_t} \left[ \left(\bs{\nabla}_X \cdot \tb{T} + \tb{b}_0 \right) \cdot \bs{v} + \tb{T}\, \bs{\colon} \bs{\dot{F}}  + \bs{\xi}_\gamma \cdot  \bs{\nabla}_X \dot{\gamma}^p + \bs{\xi}_\phi \cdot  \bs{\nabla}_X \dot{\phi} \right] dV \\ \nonumber
+  \int_{\mathcal{P}_t}  \left(\bs{\nabla}_X \cdot \bs{\xi}_\gamma \right) \dot{\gamma}^p dV  +  \int_{\mathcal{P}_t}  \left(\bs{\nabla}_X \cdot \bs{\xi}_\phi \right) \dot{\phi} dV +  \int_{\mathcal{P}_t} \left( \rho h - \bs{\nabla}_X \cdot \bs{q} \right) dV \\
\label{eqD:internal energy}
\end{eqnarray}

 Next using the macro-force balance and the micro-force balance given in Eqs. \ref{eqD:macroforce balance intVariable}, \ref{eqD:microforce balance intVariable} and \ref{eq:microscopic force balance_phi}, we may re-write Eq. \ref{eqD:internal energy} in the localized form as follows:

\begin{eqnarray}
\nonumber
\rho \dot{e} -   \zeta_\gamma \dot{\gamma}^p \, \ddot{\gamma}^p - \zeta_\phi \dot{\phi} \, \ddot{\phi}  =   \tb{T}\, \bs{\colon} \bs{\dot{F}}  +  \left(\pi_\gamma  - \bs{F}^{e\mathsf{T}} \bs{T} \bs{F}^{p\mathsf{T}} \colon \bs{N}^p \right) \dot{\gamma}^p 
+\pi_\phi \dot{\phi} \\  + \bs{\xi}_\gamma \cdot  \bs{\nabla}_X \dot{\gamma}^p + \bs{\xi}_\phi \cdot  \bs{\nabla}_X \dot{\phi}  +  \left( \rho h - \bs{\nabla}_X \cdot \bs{q} \right) 
\label{eqD:localized form of internal energy}
\end{eqnarray}

\subsection{Thermodynamic restrictions}

We now focus on the thermodynamic restrictions to be placed on the constitutive relations. These restrictions are a result of the compliance of the constitutive relations with the second law of thermodynamics, which states that the entropy production rate should be greater than or equal to zero. We enforce the requirement of the second law of thermodynamics using the local form of Clausius-Duhem inequality given as:

\begin{equation}
\rho \dot{\eta} + \bs{\nabla}_X \cdot \left(\frac{\bs{q}}{\theta} \right) - \frac{\rho h}{\theta} \geq 0
\label{eqD: Clausius-Duhem inequality}
\end{equation}
$\eta$ denotes the specific entropy. The free energy density may be written as:

\begin{equation}
{\Pi} = \Pi^e + \Pi^z + \Pi^{\bar{R}} + \Pi^\theta
\end{equation}
where, the elastic energy density constructed in Eq. \ref{eqD:energy_density_elas}  is modified  to account for ductile damage by multiplying the volumetric and the isochoric parts with the degradation functions $\mathfrak{X}_1$ and $\mathfrak{X}_2$ as follows:

\begin{equation}
\Pi^e = \mathfrak{X}_1(\phi,\gamma^p,\mathcal{J}) \Pi^e_{vol} + \mathfrak{X}_2(\phi,\gamma^p)\Pi^e_{ic}
\end{equation}

Recognizing the  tension-compression asymmetry in the damage process, we define the degradation functions, $\mathfrak{X}_1$ and $\mathfrak{X}_2$  as follows:

\begin{equation}
\mathfrak{X}_1 = \phi^{2\langle \mathcal{J}-1\rangle \mathfrak{P}_\gamma } + \eta_\phi \;\;\; \text{and} \;\;\; \mathfrak{X}_2 = \phi^{2\mathfrak{P}_\gamma }+\eta_\phi
\end{equation}
where, $\mathcal{J} := \text{det}\bs{F}$ and $\eta_\phi$ is a small parameter to prevent ill-conditioning of the elastic stiffness at the complete damage state. We define $\mathfrak{P}_\gamma = \left(\gamma^p/\gamma^p_{crit} \right)^2$ and,

\begin{align}
\langle \mathcal{J}-1\rangle = \left\{ \begin{array}{cc} 
                0 & \hspace{3mm} \text{if}\hspace{3mm} 0< \mathcal{J} < 1 \\
                1 & \hspace{3mm} \text{otherwise} \\
                \end{array} \right.
\end{align}
The free energy density may be expressed as:

\begin{equation}
\Pi = \hat{\Pi} \left( \bs{C}^e, \bs{\nabla}_{X} \gamma^p, \dot{\gamma}^p,{\gamma}^p, \Psi_t,\phi,\dot{\phi},\bs{\nabla}_{X}\phi, \theta \right)
\end{equation}

\begin{eqnarray}
\nonumber
\dot{\Pi} = \frac{\partial \Pi}{\partial \bs{C}^e} \colon \dot{\bs{C}}^e + \frac{\partial \Pi}{\partial \left(\bs{\nabla} \gamma^p \right)} \cdot \bs{\nabla} \dot{\gamma^p} + \frac{\partial \Pi}{\partial{\gamma}^p}  \dot{\gamma}^p  + \frac{\partial \Pi}{\partial \dot{\gamma}^p}  \ddot{\gamma}^p + \frac{\partial \Pi}{\partial \Psi_t}  \dot{\Psi}_t + \frac{\partial \Pi}{\partial \phi}  \dot{\phi} + \frac{\partial \Pi}{\partial \dot{\phi}}  \ddot{\phi} \\ + \frac{\partial \Pi}{\partial \left(\bs{\nabla} \phi \right)} \cdot \bs{\nabla} \dot{\phi} +\frac{\partial \Pi}{\partial \theta}  \dot{\theta}
\label{eqD:Pi_dot}
\end{eqnarray}
The second Piola-Kirchhoff kind of stress is defined as: $\bs{S}^e := 2 \frac{\partial \Pi}{\partial \bs{C}^e} $; where,

\begin{equation}
\bs{S}^e = 2k_2\mathfrak{X}_1 \left[ \left(\text{det} \bs{C}^e \right)^{\frac{1}{2}} -1 \right] \left(\text{det} \bs{C}^e\right)^{\frac{1}{2}} \bs{C}^{e-1} + 4k_3\mathfrak{X}_2 \left[\bs{C}^e - \frac{1}{3} \left(\bs{C}^e \bs{\colon} \bs{C}^e \right) \bs{C}^{e-1} \right] \left(\text{det} \bs{C}^e \right)^{-\frac{2}{3}}
\label{eqD: 2nd Piola Kirchhoff stress}
\end{equation}

Using Legendre transformation, we may write:

\begin{equation}
\dot{\Pi} = \rho \dot{e} - \rho \dot{\theta} \eta - \rho \theta \dot{\eta}
\label{eqD:Legendre transformation-2}
\end{equation}

Now using Eq. \ref{eqD:localized form of internal energy}, \ref{eqD: Clausius-Duhem inequality} in Eq. \ref{eqD:Legendre transformation-2},  we get:

\begin{eqnarray}
\nonumber
-\dot{\Pi}  -  \rho \dot{\theta} \eta  + \zeta_\gamma \,\dot{\gamma}^p\, \ddot{\gamma}^p + \zeta_\phi \dot{\phi} \, \ddot{\phi}  +  \tb{T}\, \bs{\colon} \bs{\dot{F}}  +  \left(\pi_\gamma  - \bs{F}^{e\mathsf{T}} \bs{T} \bs{F}^{p\mathsf{T}} \colon \bs{N}^p \right) \dot{\gamma}^p \\
+\pi_\phi \dot{\phi} + \bs{\xi}_\gamma \cdot  \bs{\nabla}_X \dot{\gamma}^p + \bs{\xi}_\phi \cdot  \bs{\nabla}_X \dot{\phi}  -  \frac{\bs{q}}{\theta} \cdot \bs{\nabla}_X \theta  \geq 0 
\label{eqD: second law of thermodynamics-2}
\end{eqnarray}

Next, we decompose $\pi_\gamma$ into energetic and disipative parts as follows: $\pi_\gamma=\pi^{en}_\gamma +\pi^{dis}_\gamma$. We define $\Psi_t$ in terms of the plastic deformation quantifier as follows: $\frac{k_z}{c_{el}^2 \left(1+ \Psi_t \right)^2} = \mathcal{Z} (\gamma^p, \dot{\gamma}^p, \theta) $ in Eq. \ref{eqD:Pi^z} and using Eq. \ref{eqD:Pi_dot}  in Eq. \ref{eqD: second law of thermodynamics-2},  we get:

\begin{eqnarray}
\nonumber
\left(- \bs{F}^e \bs{S}^e + \tb{T} \bs{F}^{p\mathsf{T}} \right) \colon \dot{F}^e + \left[\bs{\xi}_\gamma - \frac{\partial \Pi}{\partial \left(\bs{\nabla}_X \gamma^p \right)}  \right] \cdot \left(\bs{\nabla}_X \dot{\gamma}^p \right)+  \left(\zeta_\gamma \,\dot{\gamma}^p - \frac{\partial \Pi}{\partial \mathcal{Z}}\frac{\partial \mathcal{Z}}{\partial \dot{\gamma}^p}   \right) \ddot{\gamma}^p 
+ \pi_\gamma^{dis}  \dot{\gamma}^p  \\ \nonumber  + \left(\pi^{en}_\gamma  - \frac{\partial \Pi}{\partial \mathcal{Z}} \frac{\partial \mathcal{Z}}{\partial {\gamma}^p}  \right) \dot{\gamma}^p  + \left[\bs{\xi}_\phi - \frac{\partial \Pi}{\partial \left(\bs{\nabla}_X \phi \right)}  \right] \cdot \bs{\nabla}_X \dot{\phi}  + \left[ \pi_\phi - \frac{\partial \Pi}{\partial \phi} \right] \dot{\phi} 
\\  +  \left(\zeta_\phi \,\dot{\phi} - \frac{\partial \Pi}{\partial \dot{\phi}} \right) \ddot{\phi} 
  - \left(\frac{\partial \Pi}{\partial \theta} + \rho {\eta} \right) \dot{\theta}  - \frac{q}{\theta} \bs{\nabla}_X \theta \geq 0 
\label{eqD: Clausius-Duhem inequality-1}
\end{eqnarray}

Applying the Coleman-Noll procedure to Eq. \ref{eqD: Clausius-Duhem inequality-1}, we may have the following constitutive relations:

\begin{equation}
 \bs{F}^e \bs{S}^e = \tb{T} \bs{F}^{p\mathsf{T}} 
 \label{eqD:second piola}
\end{equation}

\begin{equation}
\bs{\xi}_\gamma = \frac{\partial \Pi}{\partial \left(\bs{\nabla}_X \gamma^p \right)}, \;\;\;\ \bs{\xi}_\phi = \frac{\partial \Pi}{\partial \left(\bs{\nabla}_X \phi \right)}, \;\;\;\       \pi^{en}_\gamma  = \frac{\partial \Pi}{\partial \mathcal{Z}} \frac{\partial \mathcal{Z}}{\partial {\gamma}^p}
\label{eqD:xi}
\end{equation}

\begin{equation}
  \zeta_\gamma = \frac{1}{\dot{\gamma}^p } \frac{\partial \Pi}{\partial \mathcal{Z}}\frac{\partial \mathcal{Z}}{\partial \dot{\gamma}^p} , \;\;\;\  \zeta_\phi = \frac{1}{\dot{\phi}} \frac{\partial \Pi}{\partial \dot{\phi}}, \;\;\;\   {\eta} = - \frac{\partial \Pi}{\rho \partial \theta}, \;\;\;\      \pi_\phi  = \frac{\partial \Pi}{\partial \phi}
 \label{eqD:micro-inertia}
\end{equation}

Thus the Clausius-Duhem inequality in Eq. \ref{eqD: Clausius-Duhem inequality-1}, takes the following form of reduced dissipation inequality:

\begin{equation}
\pi_\gamma^{dis}  \dot{\gamma}^p   - \frac{\bs{q}}{\theta} \cdot \bs{\nabla}_X \theta \geq 0
 \label{eqD:clausius-duhem inequality-2}
\end{equation}

Note that we have $\dot{\gamma}^p \geq 0; \forall \;t \in \mathbb{R}^+$, so the inequality given in Eq. \ref{eqD:clausius-duhem inequality-2} may be satisfied by setting the following: $ \pi_\gamma^{dis} = \mathcal{F}$ and $\bs{q} = - \mathcal{G} \bs{\nabla}_X \theta$, where $\mathcal{F}, \mathcal{G} \geq 0$.

\subsection{Temperature evolution equation}
 Using the constitutive relation for specific entropy given in Eq. \ref{eqD:micro-inertia}, we may write:

\begin{equation}
\dot{\eta} = - \frac{\partial^2 \Pi}{\rho \partial \theta^2} \dot{\theta}
\label{eqD:eta_dot}
\end{equation}
The specific heat constant is defined as:

\begin{equation}
\mathcal{C}_v := - \frac{\theta}{\rho} \frac{\partial^2 \Pi}{\partial \theta^2}
\label{eqD:specific heat constant}
\end{equation}
From Eqs. \ref{eqD:eta_dot} and \ref{eqD:specific heat constant}, we may write:

\begin{equation}
   \mathcal{C}_v =  \frac{\theta}{\dot{\theta}} \dot{\eta}
\label{eqD:specific heat constant-1}
\end{equation}
Using the constitutive relations in Eqs. \ref{eqD:second piola} - \ref{eqD:micro-inertia} and Eqs.  \ref{eqD:Pi_dot}, \ref{eqD:specific heat constant-1} in Eq. \ref{eqD:Legendre transformation-2}, we get:

\begin{equation}
 \dot{\theta} = \frac{1}{\rho \mathcal{C}_v} \left( \pi_\gamma^{dis} \dot{\gamma}^p + \rho h - \bs{\nabla}_X \cdot \bs{q} \right)
\label{eqD: temp. evol. equation}
\end{equation}

\subsection{Specializing of the constitutive relations}
Now we need to specialize the constitutive relations. First, we specify the temporal function $\mathcal{Z}$.  Note that $\mathcal{Z}$ acts as a handle for modelling micro-inertia for different material hardening behaviour. We consider $\mathcal{Z}$ in the power-law form given as:

\begin{equation}
\mathcal{Z} = k_\rho + k_\gamma \left( \gamma^p\right)^{1/n} \left(\dot{\bar{\gamma}}^p \right)^{-1/m} \Theta
\label{eqD:specialization of temporal function}
\end{equation}
where,

\begin{equation}
\Theta = 1- \left( \frac{\theta - \theta_{ref}}{\theta_{melt} - \theta_{ref}} \right)^{r}
\label{eqD:fun_Theta}
\end{equation}
The material constant $r$ is taken as a quadratic function of temperature, $r= \alpha_0 + \alpha_1 \theta_0 + \alpha_2 \theta_0^2$, where $\theta_0$ is the temperature at time $t_0$. The material constants $k_\rho$ and $k_\gamma$ are related to the defect or micro-inertia and the hardening behaviour of the material, respectively. We define $\dot{\bar{\gamma}}^p = \dot{\gamma}^p/\dot{\gamma}_0 $, where $\dot{\gamma}_0$ is the reference strain rate. We need to make a valid choice for $\mathcal{F}$ that does not violate  the dissipation inequality given in Eq. \ref{eqD:clausius-duhem inequality-2}. Assuming $\mathcal{F} = \left[ S_0 +  \left( \mathcal{H}_0 - \left(1-\phi \right) \mathcal{H}_s \right) \left({\gamma}^p \right)^{-1+1/n}  \right] \left(\dot{\bar{\gamma}}^p \right)^{2-1/m} \Theta$ is one of such valid choice for metal plasticity. $S_0$ and $\mathcal{H}_0$ denote the yield strength and a hardening constant, respectively. $\mathcal{H}_s$ represents a negative surface hardening constant due to dislocation absorption at the cracked surfaces satisfying the following constraint, $\mathcal{H}_s \leq \mathcal{H}_0$.

We take the material constants, $k_\rho = 0.5 \rho l_0^2$ and $k_z=-0.5 S_0 l_z^2$ \citep{gurtin2010mechanics}. Next, we correspond some of our material parameters with the material constants of a phase-field based ductile damage model. Comparing Eq. \ref{eq:microscopic force balance_phi} with the phase-field based ductile damage model, we get $k_x = G_c l_{\phi}/6 $, $k_t/c_{el} = M/(6k_\phi k_p)$, where $G_c$ and $M$ denote the critical energy release rate and a mobility constant, respectively. We use  $\bar{k}_t = -k_t/c_{el}$ and $\bar{k}_p = -k_p \sqrt{c_{el}}$. Now $k_\phi$ may be expressed as:

\begin{equation}
k_\phi = \left[\frac{G_c}{24l_\phi \left(k_x + \bar{k}_t \bar{k}_p^2 \right)} \right]^{\frac{1}{2}}
\end{equation}

\begin{figure}
\begin{subfigure}{0.45\textwidth}
\centering
\includegraphics[width=0.8\linewidth]{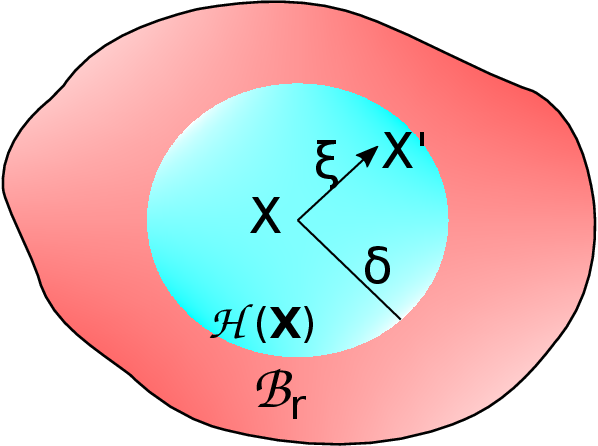}
\caption{Reference configuration}
\label{fig:PD_ref_config}
\end{subfigure}
\hfill
\begin{subfigure}{0.45\textwidth}
\centering
\includegraphics[width=0.8\linewidth]{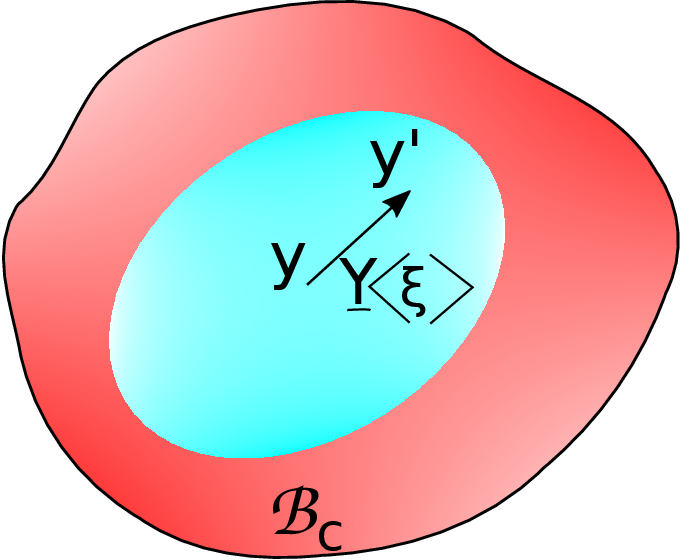}
\caption{Deformed configuration}
\label{fig:PD_curr_config}
\end{subfigure}

\caption{Schematic of PD body with reference and deformed configurations. } 
\label{fig:PD_schematic}
\end{figure}
\section{Non-ordinary state based PD formulation}
\label{vd:PD}

Keeping in mind possibly large deformation and the attendant difficulties involving extreme mesh distortions in the FEM, we intend to use the peridynamics (PD) approach for purposes of numerical implementation. We first present a brief review of the state-based peridynamics and then describe the PD based formulation for our space-time viscoplastic-damage model. State based PD is by itself a nonlocal continuum theory \citep{silling2007peridynamic}, i.e. PD considers finite distance interaction among the material points in contrast to the classical equations of motion which are based on the principle of local action. The material points interact with the other points in the neighbourhood over a finite distance ($\delta$) called the `horizon'. The horizon $\mathcal{H}(\bs{X})$ (see Fig \ref{fig:PD_schematic}) may be given as, $\mathcal{H}\left(\bs{X}\right) = \{\bs{\xi} \in \mathbb{R}^3 | \left(\bs{\xi} + \bs{X} \right) \in \mathcal{B}_r, |\bs{\xi}|<\delta\}$, where $\delta > 0$ denotes the horizon radius. Non-ordinary state based PD is the most general form of PD where the force state depends on the collective deformation of all the bonds in the horizon and its direction is not restricted to be along the direction of the bond vector --- a feature typical of bond based and ordinary state based PD. The bond vector $\bs{\xi}$ between a material point, $\bs{X} \in \mathcal{B}_r$ and its neighbour $\bs{X}^\prime \in \mathcal{B}_r$ is given as, $\bs{\xi} = \bs{X}^\prime - \bs{X}$. The deformation vector state $\barbelow{\tb{Y}}$, which maps the bond vector $\bs{\xi}$ to its deformed vector state, is defined as:

\begin{equation}
\barbelow{\tb{Y}}\left[\bs{X},t \right] \left\langle \bs{\xi} \right\rangle = \tb{y}^\prime - \tb{y}
\end{equation}
where, $\tb{y}^\prime = \bs{\rchi}(\bs{X}') \in \mathcal{B}_c$ and $\tb{y} = \bs{\rchi}(\bs{X}) \in \mathcal{B}_c$. In state based PD, the governing equations are in the integro-differential form unlike classical continuum mechanics where the PDEs cannot directly treat singularities/cracks. Since the PD formulation involves no derivatives with respect to the spatial coordinates, it allows discontinuities to be treated by its very construction. The balance of linear momentum is given as:

\begin{equation}
\rho\left(\bs{X}\right)\ddot{\tb{y}}\left(\bs{X},t\right)= \int_{\mathcal{H}\left(\bs{X}\right)}\left(\barbelow{\tb{T}}\left[\bs{X},t\right]\left\langle\bs{\xi}\right\rangle-\barbelow{\tb{T}}\left[\bs{X}',t\right]\left\langle-\bs{\xi}\right\rangle \right) dV_{\bs{X}'} + \tb{b}_0\left(\bs{X},t\right)
\label{eq:linear_momentum_PD}
\end{equation} 
where $\rho$, $\barbelow{\tb{T}}$, $\tb{b}_0$ are the mass density, force vector state and the body force density, respectively.  Balance of angular momentum implies that the following constraint is satisfied:

\begin{equation}
\int_{\mathcal{H}\left(\bs{X}\right)} \barbelow{\tb{T}}\left[\bs{X},t\right]\left\langle\bs{\xi}\right\rangle \times \barbelow{\tb{Y}}\left[\bs{X},t\right]\left\langle\bs{\xi}\right\rangle dV_{\bs{X}'}=0 
\label{eq:angular_momentum_PD}
\end{equation}

Constitutive modelling in PD should respect the constraint given in Eq \ref{eq:angular_momentum_PD}. By exploiting an equivalence of
certain quantities (e.g. strain energy density) under homogeneous deformation, a constitutive correspondence is proposed by \cite{silling2007peridynamic} to incorporate classically known constitutive models within the PD framework.

\subsection{Kinematic states and non-local approximation}

The state of the body is described by the relative displacement vector state, relative damage field scalar state, relative temperature field scalar state and relative equivalent plastic strain field scalar state. The relative displacement vector state is defined as:

\begin{equation}
\barbelow{\tb{U}}[\bs{X}]\left\langle\bs{\xi}\right\rangle=\tb{u}' - \tb{u}
\label{eq:displacement_state}
\end{equation}
Similarly, the relative damage field scalar state, the relative temperature field scalar state and the relative equivalent plastic strain field scalar state are given below in Eqs.\ref{eq:damage_field_state},\ref{eq:temperature_state} and \ref{eq:plastic_strain_state}, respectively.

\begin{equation}
\barbelow{\Phi}[\bs{X}]\left\langle\bs{\xi}\right\rangle={\phi}' - \phi
\label{eq:damage_field_state}
\end{equation}

\begin{equation}
\barbelow{\Theta}[\bs{X}]\left\langle\bs{\xi}\right\rangle={\theta}' - \theta
\label{eq:temperature_state}
\end{equation}

\begin{equation}
\barbelow{\Gamma}^p[\bs{X}]\left\langle\bs{\xi}\right\rangle={\gamma^p}' - \gamma^p
\label{eq:plastic_strain_state}
\end{equation}
where, $\left(\cdot\right)' = (\cdot)(\bs{X}')$. A nonlocal form of the deformation gradient may be given as:

\begin{equation}
\overline{\tb{F}}\left(\barbelow{\tb{Y}}\right)=\left[\int_{\mathcal{H}\left(\bs{X}\right)} \omega\left(|\bs{\xi}|\right) \left(\barbelow{\tb{Y}}\left\langle\bs{\xi}\right\rangle \otimes \bs{\xi} \right) dV_{\bs{X}'}\right] \overline{\tb{K}}^{-1}
\label{eq:nonlocal_deformation_gradient}
\end{equation}
$\omega\left(|\bs{\xi}|\right)$ denotes the influence function and the shape tensor is defined as:

\begin{equation}
\overline{\tb{K}}= \int_{\mathcal{H}\left(\bs{X}\right)} \omega\left(|\bs{\xi}|\right) \left(\bs{\xi} \otimes \bs{\xi} \right) dV_{\bs{X}'}
\label{eq:shape_tensor} 
\end{equation}

Similarly, the non-local approximation of the damage field gradient vector may be defined as:

\begin{equation}
\overline{\tb{G}}_\phi\left(\barbelow{\Phi}\right)=\left[\int_{\mathcal{H}\left(\bs{X}\right)} \omega\left(|\bs{\xi}|\right) \barbelow{\Phi}\left\langle\bs{\xi}\right\rangle \,\bs{\xi} dV_{\bs{X}'}\right] \overline{\tb{K}}^{-1}
\label{eq:nonlocal_phase_field_gradient}
\end{equation}
The non-local versions of temperature gradient and equivalent plastic strain gradient may be written as:

\begin{equation}
\overline{\tb{G}}_\theta \left(\barbelow{\Theta}\right)=\left[\int_{\mathcal{H}\left(\bs{X}\right)} \omega\left(|\bs{\xi}|\right) \barbelow{\Theta}\left\langle\bs{\xi}\right\rangle \bs{\xi} dV_{\bs{X}'}\right] \overline{\tb{K}}^{-1}
\label{eq:nonlocal_temperature_gradient}
\end{equation}

\begin{equation}
\overline{\tb{G}}_{{\gamma}}\left(\barbelow{\Gamma}^p\right)=\left[\int_{\mathcal{H}\left(\bs{X}\right)} \omega\left(|\bs{\xi}|\right) \barbelow{\Gamma}^p\left\langle\bs{\xi}\right\rangle \bs{\xi} dV_{\bs{X}'}\right] \overline{\tb{K}}^{-1}
\label{eq:nonlocal_equivalent_plastic_strain_gradient}
\end{equation}

\subsection{PD balance laws}

In addition to the well-known balances of linear and angular momenta given in Eqs. \ref{eq:linear_momentum_PD} and \ref{eq:angular_momentum_PD} respectively, we propose the following  micro-force balances for plastic deformation and damage. The micro-force balance associated with the damage field is given as:

\begin{equation}
\zeta_\phi \ddot{\phi}\left(\bs{X},t\right)=\int_{\mathcal{H}\left(\bs{X}\right)}
\left(\barbelow{\xi}_\phi\left[\bs{X},t\right]\left\langle\bs{\xi}
\right\rangle-\barbelow{\xi}_\phi\left[\bs{X}',t\right]\left\langle-
\bs{\xi}\right\rangle \right) dV_{\bs{X}'} - {\pi}_\phi\left(\bs{X},t\right)
\label{eq:micro_force_balance_damage_PD}
\end{equation}
Similarly, the micro-force balance for viscoplastic deformation is postulated as:

\begin{equation}
\zeta_\gamma \ddot{\gamma}^p\left(\bs{X},t\right)=\int_{\mathcal{H}\left(\bs{X}\right)}
\left(\barbelow{\xi}_\gamma\left[\bs{X},t\right]\left\langle\bs{\xi}\right
\rangle-\barbelow{\xi}_\gamma\left[\bs{X}',t\right]\left\langle-\bs{\xi}\right\rangle \right) dV_{\bs{X}'} - \bar{\pi}_\gamma\left(\bs{X},t\right)
\label{eq:micro_force_balance_plasticity_PD}
\end{equation}
where, $\barbelow{\xi}_\phi$ and $\barbelow{\xi}_\gamma$ represent the micro-force scalar state associated with the damage field and viscoplastic deformation, respectively. The driving force for plastic deformation is given as: $\bar{\pi}_\gamma\left(\bs{X},t\right) =  \pi_\gamma  - \left(\overline{\bs{F}}^{e\mathsf{T}} {\overline{\tb{T}}} {\overline{\bs{F}}^{p\mathsf{T}}} \right)   \colon \overline{\bs{N}}^p$.  We now show that the micro-force balance associated with the damage field and viscoplastic deformation  given in Eqs. \ref{eq:micro_force_balance_damage_PD} and \ref{eq:micro_force_balance_plasticity_PD}  respectively are globally satisfied.

\begin{prop}
	If the forces arising out of damage field evolution Eq.\ref{eq:micro_force_balance_damage_PD} for a bounded body $\mathcal{B}_r$ holds for any $\tb{X} \in$ $\mathcal{B}_r$ then the corresponding micro-force balance is satisfied globally.
\end{prop}

\begin{proof}
Integrating Eq.\ref{eq:micro_force_balance_damage_PD} over the entire body $\mathcal{B}_r$, we obtain:
	
\begin{equation}
\int_{\mathcal{B}_r} \zeta_\phi \ddot{\phi}\left(\bs{X},t\right) dV_{\bs{X}} = \int_{\mathcal{B}_r} \int_{ \mathcal{H} \left( \bs{X} \right) } \left( \barbelow{\xi}_\phi \left[ \bs{X}, t \right] \left\langle \bs{\xi} \right\rangle - \barbelow{\xi}_\phi \left[ \bs{X}', t \right] \left\langle - \bs{\xi} \right\rangle \right) dV_{\bs{X}'} dV_{\bs{X}} - \int_{\mathcal{B}_r} {\pi}_\phi \left(\bs{X},t\right) dV_{\bs{X}}
\label{eq:micro_force_balance_phase_field_PD_1}
\end{equation}
	
Note that the interaction with a material point $\bs{X}$ vanishes outside the horizon $\mathcal{H} \left( \bs{X} \right)$, so we may extend the inner integral in Eq.\ref{eq:micro_force_balance_phase_field_PD_1} over the entire body $\mathcal{B}_r$ as:	

\begin{equation}
\int_{\mathcal{B}_r} \zeta_\phi \ddot{\phi}\left(\bs{X},t\right) dV_{\bs{X}} = \int_{\mathcal{B}_r} \int_{ \mathcal{B}_r } \left( \barbelow{\xi}_\phi  - \barbelow{\xi}'_\phi \right) dV_{\bs{X}'} dV_{\bs{X}} - \int_{\mathcal{B}_r} {\pi}_\phi \left(\bs{X},t\right) dV_{\bs{X}}
\label{eq:micro_force_balance_phase_field_PD_2}
\end{equation}
We have used the notations $\barbelow{\xi}_\phi = \barbelow{\xi}_\phi \left[ \bs{X}, t \right] \left\langle \bs{\xi} \right\rangle $ and $\barbelow{\xi}_{\phi}' = \barbelow{\xi}_\phi \left[ \bs{X}', t \right] \left\langle - \bs{\xi} \right\rangle$ in Eq. \ref{eq:micro_force_balance_phase_field_PD_2}. Now applying the change of variable $\bs{X} \leftrightarrow \bs{X}^{'}$, the first integral on the right hand side of Eq. \ref{eq:micro_force_balance_phase_field_PD_2} may be written as:

\begin{equation}
\int_{\mathcal{B}_r} \int_{\mathcal{B}_r} \left( \barbelow{\xi}_\phi - \barbelow{\xi}_{\phi}' \right) dV_{\bs{X}'} dV_{\bs{X}} = \int_{\mathcal{B}_r} \int_{ \mathcal{B}_r} \left(\barbelow{\xi}_{\phi}' - \barbelow{\xi}_\phi \right) dV_{\bs{X}} dV_{\bs{X}^{'}} = - \int_{\mathcal{B}_r} \int_{ \mathcal{B}_r} \left(\barbelow{\xi}_\phi - \barbelow{\xi}_{\phi}' \right) dV_{\bs{X}'} dV_{\bs{X}} = 0
\label{eq:micro_force_balance_phase_field_PD_3}
\end{equation}
Thus, Eq. \ref{eq:micro_force_balance_phase_field_PD_2} simplifies as, $\int_{\mathcal{B}_r} \zeta_\phi \ddot{\phi}\left(\bs{X},t\right) dV_{\bs{X}} + \int_{\mathcal{B}_r} {\pi}_\phi \left(\bs{X},t\right) dV_{\bs{X}}$ = 0, which implies that the micro-force balance given in Eq. \ref{eq:micro_force_balance_damage_PD} is satisfied globally. Similarly one may show that the micro-force balance given in Eq. \ref{eq:micro_force_balance_plasticity_PD} is also globally satisfied. 	
\end{proof}

\subsection{PD energy balance and constitutive model}

We present the PD energy balance equation and describe the constitutive correspondence with the classical material model \citep{silling2010peridynamic}. First, we discuss the non-local PD energy balance equations and derive the internal energy evolution equation in the PD setup.

\subsubsection{PD energy balance}

The external power may be written using the non-local PD states as follows:

\begin{align}
\nonumber
\mathcal{P}^{ext} = \int_{\mathcal{P}_t} \int_{\mathcal{B}_r\backslash \mathcal{P}_t} \left( \barbelow{\tb{T}} \cdot \dot{\tb{y}}' - \barbelow{\tb{T}}' \cdot \dot{\tb{y}}\right) dV_{\bs{X}'} dV_{\bs{X}} + \int_{\mathcal{P}_t} \int_{\mathcal{B}_r\backslash \mathcal{P}_t} \left(\barbelow{\xi}_\gamma \dot{\gamma^p}' - \barbelow{\xi}_\gamma' \dot{\gamma}^p\right) dV_{\bs{X}'} dV_{\bs{X}} \\ 
 + \int_{\mathcal{P}_t} \int_{\mathcal{B}_r\backslash  \mathcal{P}_t} \left(\barbelow{\xi}_\phi {\dot{\phi}}' - \barbelow{\xi}_\phi' \dot{\phi} \right) dV_{\bs{X}'} dV_{\bs{X}} + \int_{\mathcal{P}_t} \tb{b}_0 \left(\bs{X},t \right) \cdot \dot{\tb{y}} dV_{\bs{X}} 
\label{eq:external_power_PD}
\end{align}
The rate of heat supply can be expressed as:

\begin{equation}
\mathcal{R} =  \int_{\mathcal{P}_t} \rho h\, dV_{\bs{X}} - \int_{\mathcal{P}_t} \int_{\mathcal{B}_r\backslash \mathcal{P}_t} \left(\barbelow{q} - \barbelow{q}' \right) dV_{\bs{X}'} dV_{\bs{X}} 
\label{eq:thermal_power_PD}
\end{equation}
Using Eqs. \ref{eqD: dK/dt}, \ref{eqD:du/dt} and Eqs. \ref{eq:external_power_PD}, \ref{eq:thermal_power_PD} in Eq. \ref{eqD:First law in rate form}, one may write the PD energy balance equation in the rate form as follows:

\begin{align}
\nonumber
\int_{\mathcal{P}_t} \rho \dot{e} dV_{\bs{X}} +  \int_{\mathcal{P}_t}  \rho \, \ddot{\tb{y}} \cdot \dot{\tb{y}} dV_{\bs{X}} =  \int_{\mathcal{P}_t} \int_{\mathcal{B}_r\backslash \mathcal{P}_t} \left( \barbelow{\tb{T}} \cdot \dot{\tb{y}}' - \barbelow{\tb{T}}' \cdot \dot{\tb{y}}\right) dV_{\bs{X}'} dV_{\bs{X}} + \int_{\mathcal{P}_t} \tb{b}_0 \left(\bs{X},t \right)\cdot\dot{\tb{y}}dV_{\bs{X}}\\ \nonumber
 + \int_{\mathcal{P}_t} \int_{\mathcal{B}_r\backslash  \mathcal{P}_t} \left(\barbelow{\xi}_\phi {\dot{\phi}}' - \barbelow{\xi}_\phi' \dot{\phi} \right) dV_{\bs{X}'} dV_{\bs{X}} + \int_{\mathcal{P}_t} \int_{\mathcal{B}_r\backslash \mathcal{P}_t} \left(\barbelow{\xi}_\gamma \dot{\gamma^p}' - \barbelow{\xi}_\gamma' \dot{\gamma}^p\right) dV_{\bs{X}'} dV_{\bs{X}}  \\
 + \int_{\mathcal{P}_t} \rho h\, dV_{\bs{X}} - \int_{\mathcal{P}_t} \int_{\mathcal{B}_r\backslash \mathcal{P}_t} \left(\barbelow{q} - \barbelow{q}' \right) dV_{\bs{X}'} dV_{\bs{X}} 
\label{eq:internal_energy_PD1}
\end{align}
The following identities hold from the anti-symmetry of the integrand:

\begin{equation}
\int_{\mathcal{P}_t} \int_{\mathcal{P}_t} \left(\barbelow{\tb{T}} \cdot \dot{\tb{y}}' - \barbelow{\tb{T}}' \cdot \dot{\tb{y}} \right) dV_{\bs{X}'}dV_{\bs{X}} = 0 
\label{eq:identity_1}
\end{equation}
\begin{equation}
\int_{\mathcal{P}_t}\int_{\mathcal{P}_t}\left(\barbelow{\xi}_\phi {\dot{\phi}}'-\barbelow{\xi}_\phi' \dot{\phi}\right) dV_{\bs{X}'}dV_{\bs{X}}=0 
\label{eq:identity_2}
\end{equation}

\begin{equation}
\int_{\mathcal{P}_t} \int_{\mathcal{P}_t} \left( \barbelow{\xi}_\gamma \dot{\gamma^p}' - \barbelow{\xi}_\gamma' \dot{\gamma}^p \right) dV_{\bs{X}'}dV_{\bs{X}} = 0 
\label{eq:identity_3}
\end{equation}
\begin{equation}
\int_{\mathcal{P}_t}\int_{\mathcal{P}_t}\left(\barbelow{q}-\barbelow{q}' \right) dV_{\bs{X}'}dV_{\bs{X}} = 0
\label{eq:identity_4}
\end{equation}

Using the identities given in Eqs. \ref{eq:identity_1}-\ref{eq:identity_4}, we may recast Eq. \ref{eq:internal_energy_PD1} as:

\begin{align}
\nonumber
\int_{\mathcal{P}_t} \rho \dot{e} dV_{\bs{X}} +  \int_{\mathcal{P}_t}  \rho \, \ddot{\tb{y}} \cdot \dot{\tb{y}} dV_{\bs{X}} =  \int_{\mathcal{P}_t} \int_{\mathcal{B}_r} \left( \barbelow{\tb{T}} \cdot \dot{\tb{y}}' - \barbelow{\tb{T}}' \cdot \dot{\tb{y}}\right) dV_{\bs{X}'} dV_{\bs{X}} + \int_{\mathcal{P}_t} \tb{b}_0 \left(\bs{X},t \right)\cdot\dot{\tb{y}} dV_{\bs{X}}  \\ \nonumber
 + \int_{\mathcal{P}_t} \int_{\mathcal{B}_r} \left(\barbelow{\xi}_\phi {\dot{\phi}}' - \barbelow{\xi}_\phi' \dot{\phi} \right) dV_{\bs{X}'} dV_{\bs{X}} + \int_{\mathcal{P}_t} \int_{\mathcal{B}_r} \left(\barbelow{\xi}_\gamma \dot{\gamma^p}' - \barbelow{\xi}_\gamma' \dot{\gamma}^p\right) dV_{\bs{X}'} dV_{\bs{X}}  \\
 + \int_{\mathcal{P}_t} \rho h dV_{\bs{X}} - \int_{\mathcal{P}_t} \int_{\mathcal{B}_r} \left(\barbelow{q} - \barbelow{q}' \right) dV_{\bs{X}'} dV_{\bs{X}} 
\label{eq:internal_energy_PD2}
\end{align}

Note the following identities,

\begin{equation}
\left( \barbelow{\tb{T}} - \barbelow{\tb{T}}' \right) \cdot \dot{\tb{y}} =  \barbelow{\tb{T}} \cdot \left( \dot{\tb{y}} - \dot{\tb{y}}' \right) + \left( \barbelow{\tb{T}} \cdot \dot{\tb{y}}' - \barbelow{\tb{T}}' \cdot \dot{\tb{y}}\right) 
\label{eq:force_state_identity}
\end{equation}
\begin{equation}
\left( \barbelow{\xi}_\phi - \barbelow{\xi}_\phi' \right) \dot{\phi} =  \barbelow{\xi}_\phi \left( \dot{\phi} - \dot{\phi}' \right) + \left( \barbelow{\xi}_\phi {\dot{\phi}}' - \barbelow{\xi}_\phi' \dot{\phi} \right) 
\label{eq:damage_field_identity}
\end{equation}
\begin{equation}
\left( \barbelow{\xi}_\gamma - \barbelow{\xi}_\gamma' \right) \dot{\gamma}^p =  \barbelow{\xi}_\gamma \left( \dot{\gamma}^p - \dot{\gamma^p}' \right) + \left( \barbelow{\xi}_\gamma {\dot{\gamma^p}}' - \barbelow{\xi}_\gamma' \dot{\gamma}^p \right) 
\label{eq:plasticity_identity}
\end{equation}
Using the identities given in Eqs. \ref{eq:force_state_identity}-\ref{eq:plasticity_identity} and substituting the PD balance laws in Eq. \ref{eq:internal_energy_PD2}, the localized form of the internal energy density evolution equation may be expressed as:

\begin{align}
\nonumber
\rho \dot{e}  = \int_{\mathcal{B}_r} \barbelow{\tb{T}} \cdot \left( \dot{\tb{y}}' - \dot{\tb{y}} \right) dV_{\bs{X}'} + \int_{\mathcal{B}_r} \barbelow{\xi}_\gamma  \left( \dot{\gamma^p}' - \dot{\gamma}^p \right) dV_{\bs{X}'} + \int_{\mathcal{B}_r} \barbelow{\xi}_\phi \left( \dot{\phi}' - \dot{\phi} \right) dV_{\bs{X}'} \\  
+\left( \zeta_\gamma  \, \ddot{\gamma}^p + \bar{\pi}_\gamma \right) \dot{\gamma}^p   +  \left( \zeta_\phi \, \ddot{\phi} + {\pi}_\phi \right)\dot{\phi}  + \rho h  -  \int_{\mathcal{B}_r} \left(\barbelow{q} - \barbelow{q}' \right) dV_{\bs{X}'}
\label{eq:internal_energy_rate_local}
\end{align}

\subsubsection{Constitutive correspondence}

First, we rewrite the evolution equation for internal energy density derived in Eq. \ref{eqD:localized form of internal energy} by replacing the classical gradients with their corresponding non-local PD approximations as follows:

\begin{align}
\nonumber
\rho \dot{e} = \overline{\tb{T}}: \dot{\overline{\tb{F}}} + \overline{\bs{\xi}}_\gamma \cdot \dot{\overline{\tb{G}}}_{{\gamma}} + \overline{\bs{\xi}}_\phi \cdot \dot{\overline{\tb{G}}}_\phi + \left( \zeta_\gamma  \, \ddot{\gamma}^p +\bar{\pi}_\gamma  \right) \dot{\gamma}^p + \left( \zeta_\phi \, \ddot{\phi}  + {\pi}_\phi \right) \dot{\phi}  + \rho h  \\
-\int_{\mathcal{B}_r} \left(\barbelow{q} - \barbelow{q}' \right) dV_{\bs{X}'}
\label{eq:nonlocal_internal_power}
\end{align}
where, $\overline{\tb{T}}, \overline{\bs{\xi}}_\gamma$ and $\overline{\bs{\xi}}_\phi$ are the non-local approximations of $\tb{T}, \bs{\hat{\xi}}_\gamma$ and $\bs{\hat{\xi}}_\phi$ respectively. Substituting the non-local gradients given in Eqs. \ref{eq:nonlocal_deformation_gradient}-\ref{eq:nonlocal_equivalent_plastic_strain_gradient} in Eq. \ref{eq:nonlocal_internal_power}, we get:

\begin{align}
\nonumber
\rho \dot{e} = \overline{\tb{T}}:  \left[ \int_{\mathcal{H}} \omega \left( \barbelow{\dot{\tb{Y}}} \left\langle \bs{\xi} \right\rangle \otimes  \bs{\xi} \right) dV_{\bs{X}'} \right] \overline{\tb{K}}^{-1} + \overline{\bs{\xi}}_\gamma \cdot \left[ \int_{\mathcal{H}} \omega \barbelow{\dot{\Gamma}}^p \left\langle \bs{\xi} \right\rangle \bs{\xi}  dV_{\tb{X}'} \right] \overline{\tb{K}}^{-1} -\int_{\mathcal{B}_r} \left(\barbelow{q} - \barbelow{q}' \right) dV_{\bs{X}'} \\ 
 + \overline{\bs{\xi}}_\phi \cdot \left[ \int_{\mathcal{H}} \omega \barbelow{\dot{\Phi}} \left\langle \bs{\xi} \right\rangle \bs{\xi}  dV_{\tb{X}'} \right] \overline{\tb{K}}^{-1} + \left( \zeta_\gamma  \, \ddot{\gamma}^p +\bar{\pi}_\gamma  \right) \dot{\gamma}_p + \left( \zeta_\phi \, \ddot{\phi}  + {\pi}_\phi \right) \dot{\phi}  + \rho h  
\label{eq:nonlocal_internal_power_1}
\end{align}

Note that, since $\overline{\tb{T}}, \overline{\bs{\xi}}_\gamma$ and $\overline{\bs{\xi}}_\phi$ are functions of $\bs{X}$ only, we may recast Eq. \ref{eq:nonlocal_internal_power_1} as:

\begin{align}
\nonumber
\rho \dot{e} = \int_{\mathcal{H}} \left( \omega \overline{\tb{T}} \overline{\tb{K}}^{-1} \bs{\xi} \right) \cdot  \barbelow{\dot{\tb{Y}}}  \left\langle \bs{\xi} \right\rangle  dV_{\tb{X}'} + \int_{\mathcal{H}}  \left( \omega \overline{\bs{\xi}}_\gamma \cdot \left( \bs{\xi} \overline{\tb{K}}^{-1} \right) \right) \cdot  \barbelow{\dot{\Gamma}}_p \left\langle \bs{\xi} \right\rangle dV_{\tb{X}'} -\int_{\mathcal{B}_r} \left(\barbelow{q} - \barbelow{q}' \right) dV_{\bs{X}'}  \\ 
+ \int_{\mathcal{H}}  \left( \omega \overline{\bs{\xi}}_\phi \cdot \left( \bs{\xi} \overline{\tb{K}}^{-1} \right) \right) \cdot  \barbelow{\dot{\Phi}}  \left\langle \bs{\xi} \right\rangle dV_{\bs{X}'} + \left( \zeta_\gamma  \, \ddot{\gamma}^p +\bar{\pi}_\gamma  \right) \dot{\gamma}^p + \left( \zeta_\phi \, \ddot{\phi}  +  {\pi}_\phi \right) \dot{\phi}  + \rho h  
\label{eq:nonlocal_internal_power_2}
\end{align}
Comparing Eqs. \ref{eq:internal_energy_rate_local} and \ref{eq:nonlocal_internal_power_2}, we get:

\begin{equation}\groupequation{
\begin{split}
&\barbelow{\tb{T}} = \omega \overline{\tb{T}} \overline{\tb{K}}^{-1} \bs{\xi}\\
&\barbelow{\xi}_\phi = \omega \overline{\bs{\xi}}_\phi \cdot \left( \bs{\xi} \overline{\tb{K}}^{-1}\right)\\
&\barbelow{\xi}_\gamma = \omega \overline{\bs{\xi}}_\gamma \cdot \left( \bs{\xi} \overline{\tb{K}}^{-1} \right)
\end{split}
}\end{equation}

Similarly, based on entropy equivalence and using Eq. \ref{eqD: Clausius-Duhem inequality}, we get:

\begin{equation}
\barbelow{q} = \omega \frac{\theta^\prime}{\theta} \overline{\bs{q}} \cdot \left(\bs{\xi} \overline{\tb{K}}^{-1}\right)
\end{equation}
where, $\overline{\bs{q}}$ represents the non-local approximation of $\bs{q}$. We adopt a rational bond breaking criterion (see \cite{roy2017peridynamics}) in the PD formulation. Following \cite{tupek2013approach}, we split influence function $\omega\left(|\bs{\xi}|\right)$ as follows:

\begin{equation}
\omega\left(|\bs{\xi}|, \phi, \phi^\prime \right) = \overline{\omega}\left(|\bs{\xi}|\right) \hat{\omega} \left( \phi, \phi^\prime \right)
\end{equation}
where, $ \overline{\omega}\left(|\bs{\xi}|\right)$ represents the influence function in the undamaged material and we define $\hat{\omega} \left( \phi, \phi^\prime \right)$ as:

\begin{align}
\hat{\omega} \left( \phi, \phi^\prime \right) =  \left\{ \begin{array}{cc} 
                0 &\text{for}\hspace{1mm} \frac{\phi+\phi^\prime}{2} = 0 \hspace{3mm}\text{and} \hspace{3mm} \lambda > 0 \\
                1 &  \text{otherwise} \\
                \end{array} \right.
\end{align}
The bond stretch $(\lambda)$ is defined as:
\begin{equation}
\lambda = \frac{|\barbelow{\tb{Y}}\left\langle\bs{\xi}\right\rangle|-|\bs{\xi}|}{|\bs{\xi}|}
\end{equation}


{\renewcommand{\arraystretch}{1.0}
\begin{table}[th]
\centering
\caption{Material Parameters for OFHC copper and titanium (TA15) alloy.} 
\begin{tabular}{|l c c |c c c|}
 \hline \hline			
Parameters & OHFC & TA15  &  Parameters &   OHFC  & TA15 \\
\hline	
 $\mu$ \; (GPa) & 46.16 &  48.82      & $\nu$   & 0.3      &  0.28 \\
 $\rho$ \; (kg/m$^3$)  & 8960 &  4560       & n     & 0.759      &  0.80 \\
 $S_0$ \; (MPa)    & 35    & 50 & m     & 0.524       &  0.553\\ 
 $H_0$ \; (MPa)   & 580     & 150  & $\alpha_0$ & -  &  22.295\\
 $K_{Ic}$ (MPa $\sqrt{\text{m}})$ & 50e6  &  $60e6$  & $\alpha_1$   & -  & -0.0461  \\
 $H_s$ \; (MPa)   & 100     & 60  & $\alpha_2$  & -    &  $2.40e{-5}$\\ 
 $\dot{\gamma}_0$  $(s^{-1})$ & 1 & 10  & $\bar{k}_t$ $(\text{J/m})$ & 25e3   & $5.84e{3}$ \\
  $\mathcal{C}_v$ \; $(\text{J{kg}}^{-1}\text{K}^{-1})$ & 385 & 565  &    $\bar{k}_p$   & $1.73e{-9}$      &  $8.63e{-8}$ \\
 $\theta_{ref}$ \; (K) & 77 & 298    &     $l_0$ (m)  & $1e{-4}$       &  $1e{-4}$ \\
 $\theta_{melt}$ \; (K)  & 1350    & 1940  &  r & 0.22 & -    \\
  \hline  
\end{tabular}
\label{table:Material Parameters OFHC_TA15}
\end{table}}


\section{Illustrative numerical results}
\label{vd:numericalR}

We may now evaluate the predictive performance of the proposed model using benchmark problems that are of general interest. First, we demonstrate the response of a ductile material specimen in a uniaxial simulation and validate it against experimental data. Following this, we showcase simulation results for an asymmetrically notched specimen under tension and the cup-cone fracture in a dog-bone shaped round bar. We assume that adiabatic conditions prevail, so the divergence term (heat fluxes) in the temperature evolution equation vanishes. \\


\subsection{OFHC copper}

First, we present the stress-strain plot for the viscoplastic deformation and its validation against the experimental data reported by \cite{nemat1998flow}. Followed by this, a loading-unloading plot showing ductile damage is presented. Material parameters for OFHC copper used in the numerical simulation are presented in Table \ref{table:Material Parameters OFHC_TA15}. Fig. \ref{fig:OFHC_4000} shows the viscoplastic response of OFHC copper for different initial thermodynamic temperatures and a strain rate of 4000/s. Whereas the stress-strain plot for different strain rates at room temperature is plotted in Fig. \ref{fig:OFHC_296K}. The stress-strain plots show good agreement with the experimental data. We report the ductile damage response in Fig. \ref{fig:Uni_Stress}. Arrow marks on the stress-strain plot indicate the loading-unloading and reloading paths. It is worth noting that when unloading is performed, unlike brittle damage, residual strain is present as expected; this is a typical feature of ductile damage. Fig. \ref{fig:Uni_phi} shows the plot of the gauge or damage field ($\phi$) vs. strain. Gauge field ($\phi$) values of 1 and 0  respectively represent undamaged and completely damaged states. Values between 1 and 0 represent partial damage states.

\begin{figure}

\begin{subfigure}{0.45\textwidth}
\centering
\includegraphics[width=0.95\linewidth]{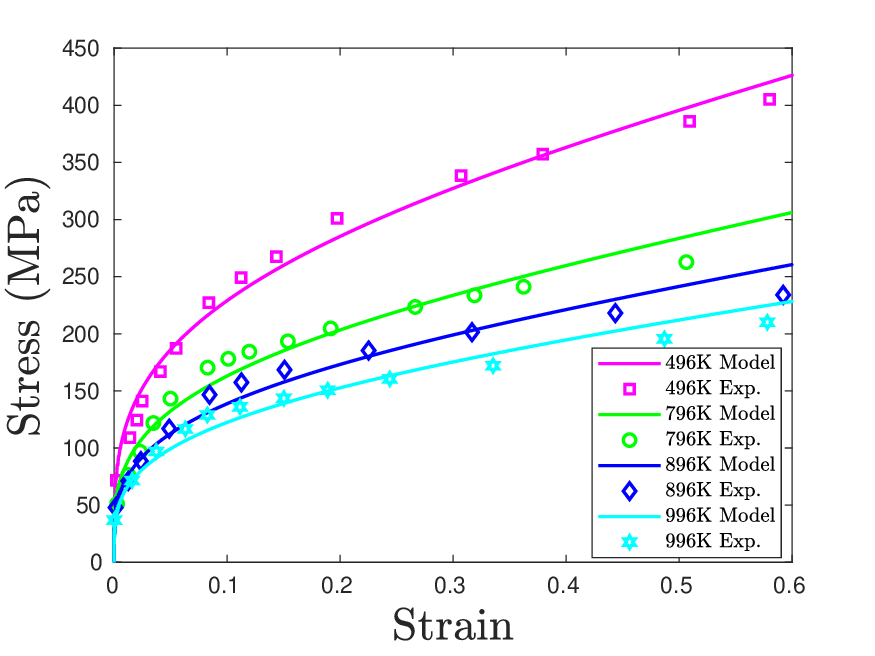}
\caption{}
\label{fig:OFHC_4000}
\end{subfigure}
\hfill
\begin{subfigure}{0.45\textwidth}
\centering
\includegraphics[width=0.95\linewidth]{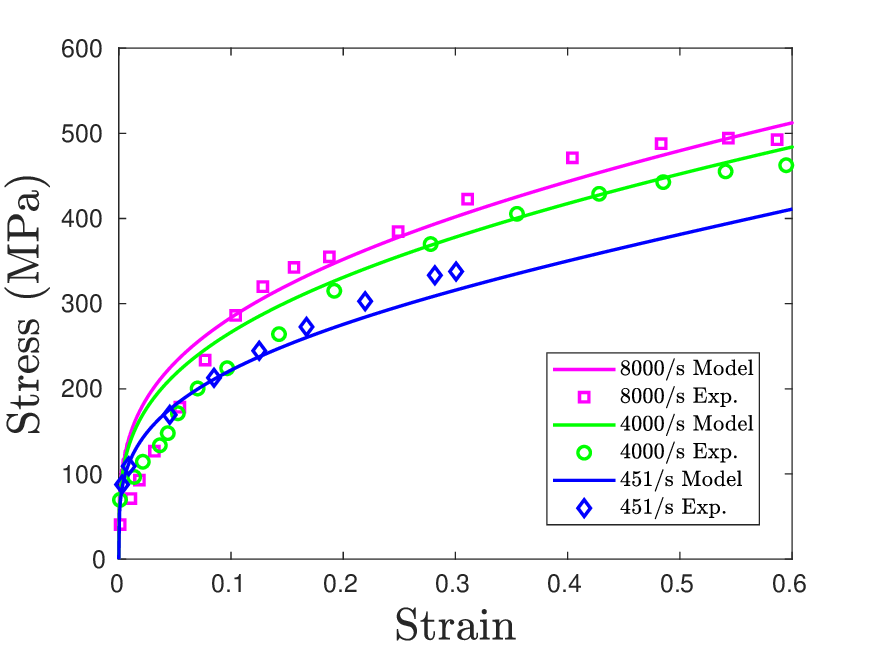}
\caption{}
\label{fig:OFHC_296K}
\end{subfigure}

\begin{subfigure}{0.45\textwidth}
\centering
\includegraphics[width=0.95\linewidth]{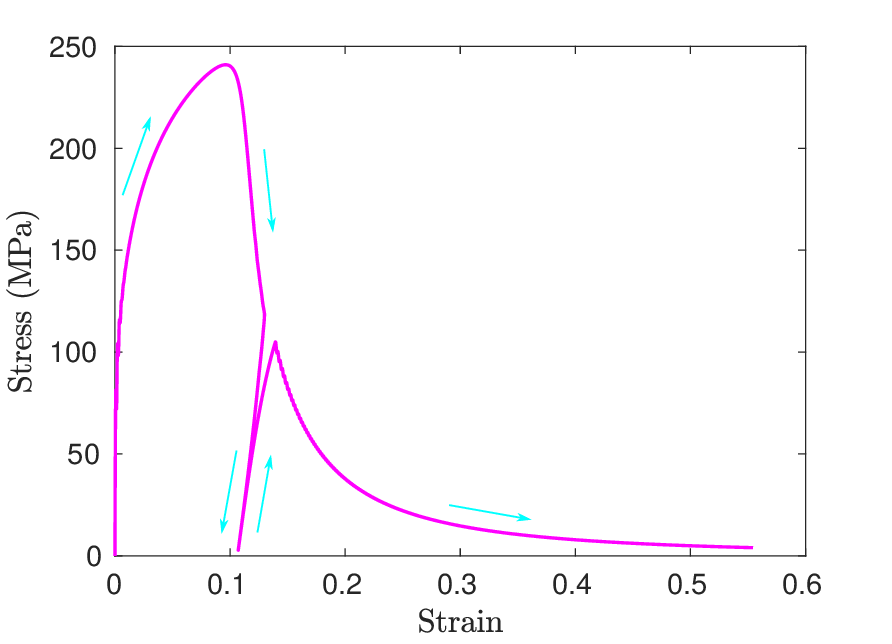}
\caption{}
\label{fig:Uni_Stress}
\end{subfigure}
\hfill
\begin{subfigure}{0.45\textwidth}
\centering
\includegraphics[width=0.95\linewidth]{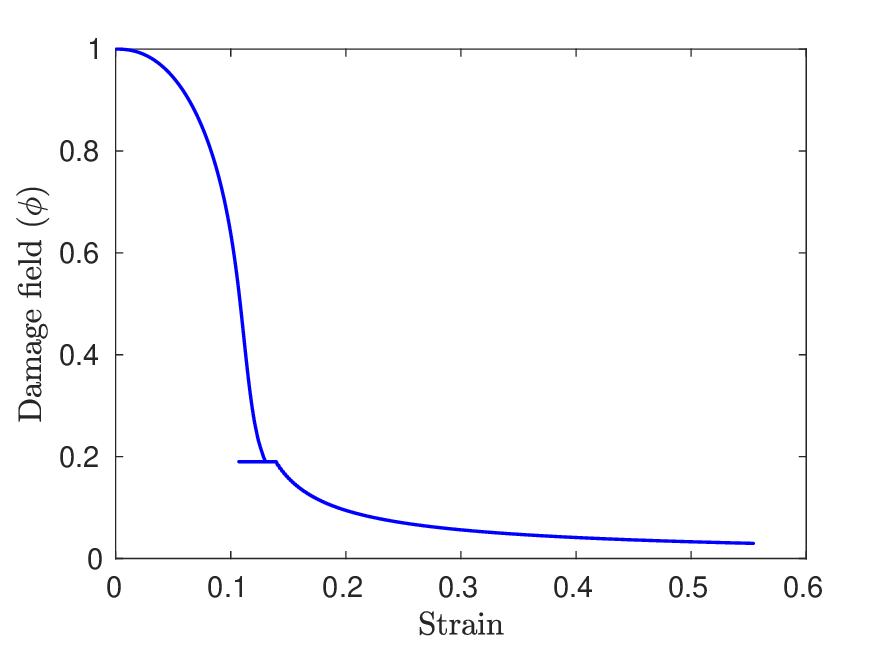}
\caption{}
\label{fig:Uni_phi}
\end{subfigure}
\caption{Model prediction of typical viscoplastic ductile damage response of OFHC copper. Figure (a) and (b) show stress-strain plot at $\phi = 1$ for strain rate = 4000/s and an initial temperature of 296K, respectively (c) Loading-unloading plot for strain rate = 8000/s  (d) damage field (gauge filed, $\phi$) vs. strain plot. Arrow marks in figure (c) indicate loading-unloading and reloading path.}
\label{fig:Uniaxial_Stress_phi_plot}
\end{figure}

\subsection{Titanium alloy (TA15)}

We present stress-strain response of titanium alloy (TA15) under a uniaxial tension test and the validation of the simulated results against the experimental data reported by  \cite{yang2015behavior}. The material parameters used in numerical simulations of TA15 are given in Table \ref{table:Material Parameters OFHC_TA15}. The numerical results for TA15 alloy at different strain rates and temperatures are shown in Fig. \ref{fig:Uniaxial_TA15}. The numerical results show good agreement with the experimental data. Stress-strain plots for strain rates 0.1/s, 0.01/s and 0.001/s with an initial temperature at 1023.15K are shown in Fig. \ref{fig:TA15_1023K}. Figs. \ref{fig:TA15_0.01} and \ref{fig:TA15_0.001} show the stress-strain response of TA15 alloy for different initial temperatures at strain rates of 0.01/s and 0.001/s, respectively. The proposed model captures combined softening due to temperature and ductile damage.

\begin{figure}
\begin{subfigure}{0.95\textwidth}
\centering
\includegraphics[width=0.5\linewidth]{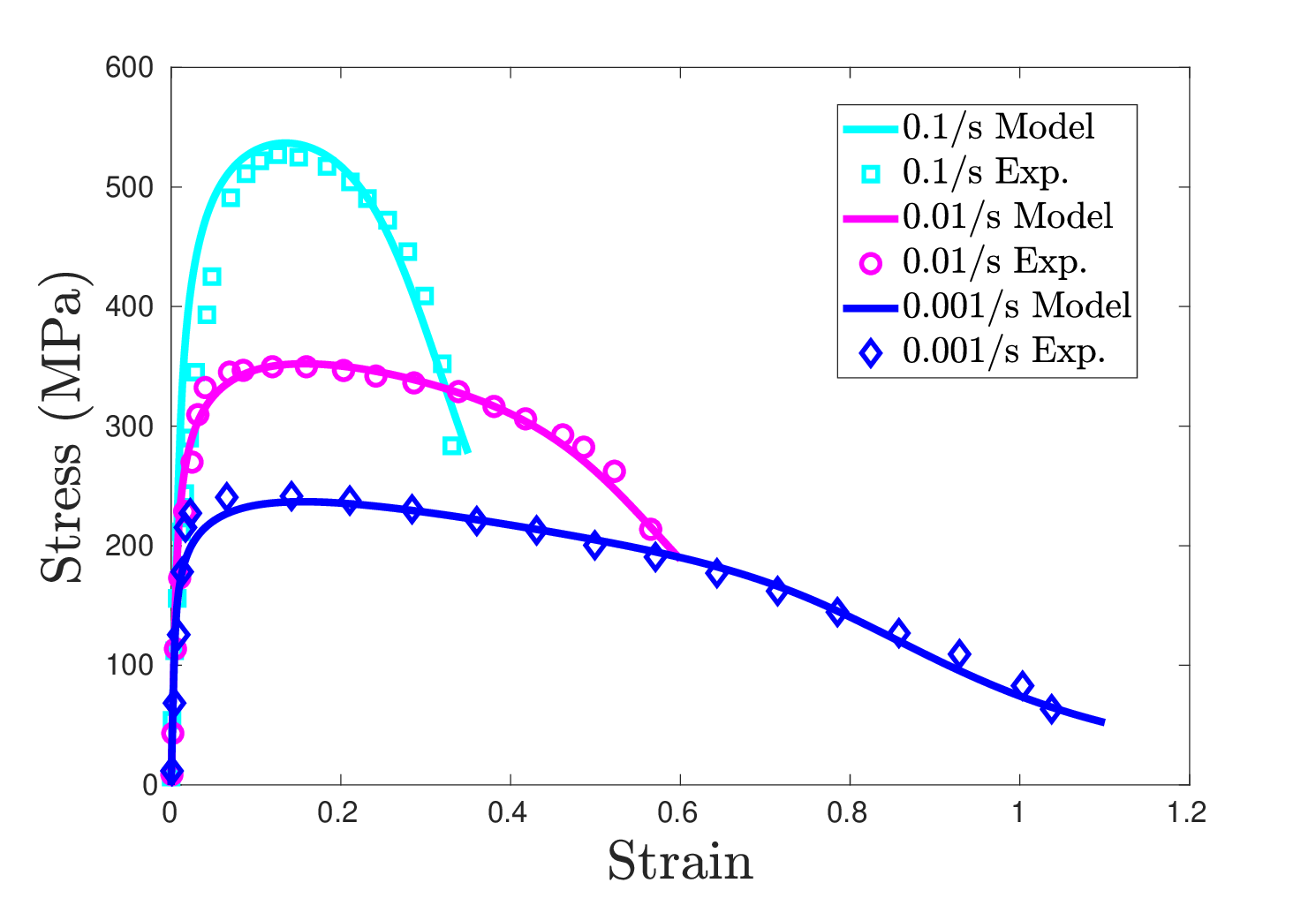}
\caption{}
\label{fig:TA15_1023K}
\end{subfigure}

\begin{subfigure}{0.4\textwidth}
\centering
\includegraphics[width=1\linewidth]{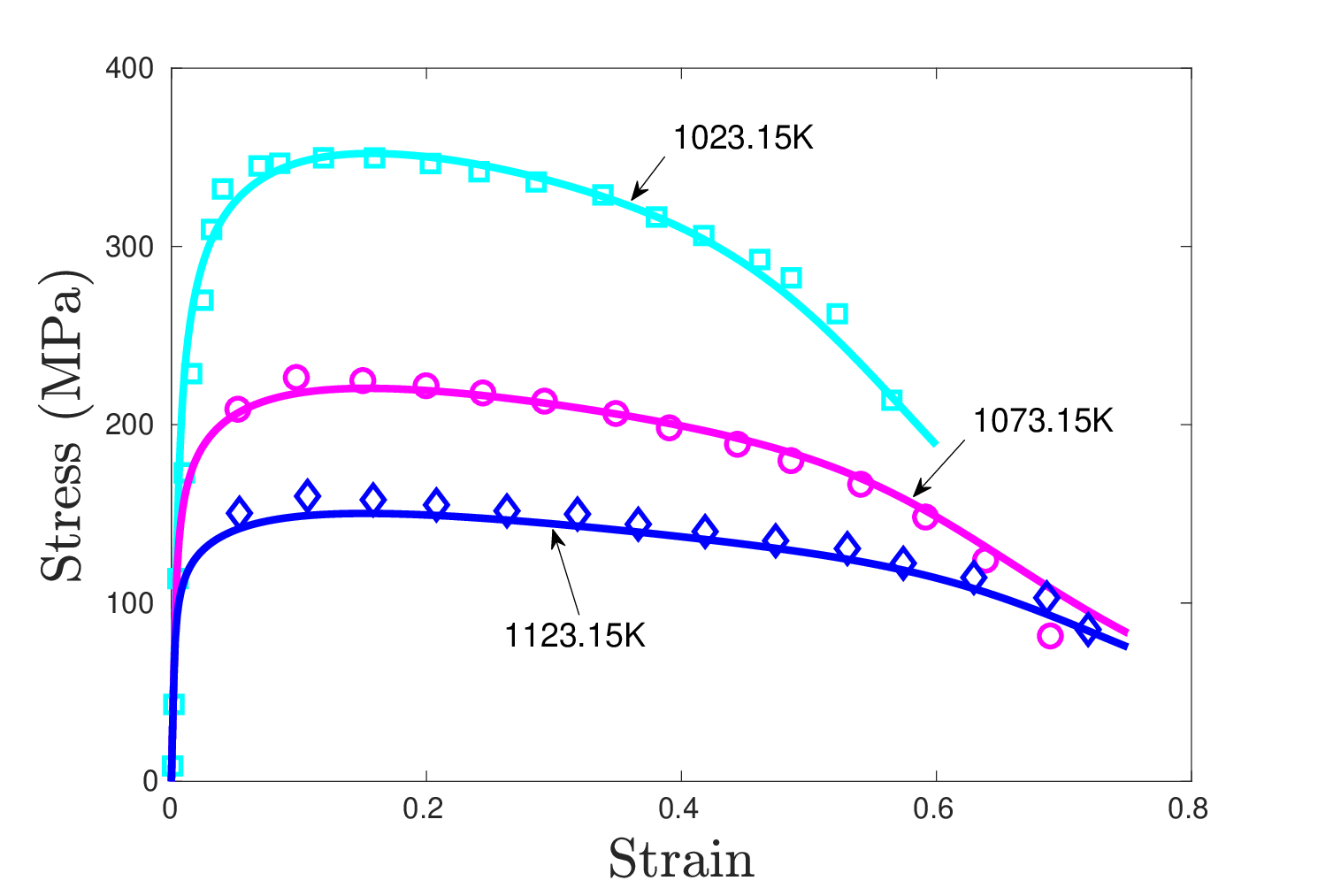}
\caption{}
\label{fig:TA15_0.01}
\end{subfigure}
\hfill
\begin{subfigure}{0.4\textwidth}
\centering
\includegraphics[width=1\linewidth]{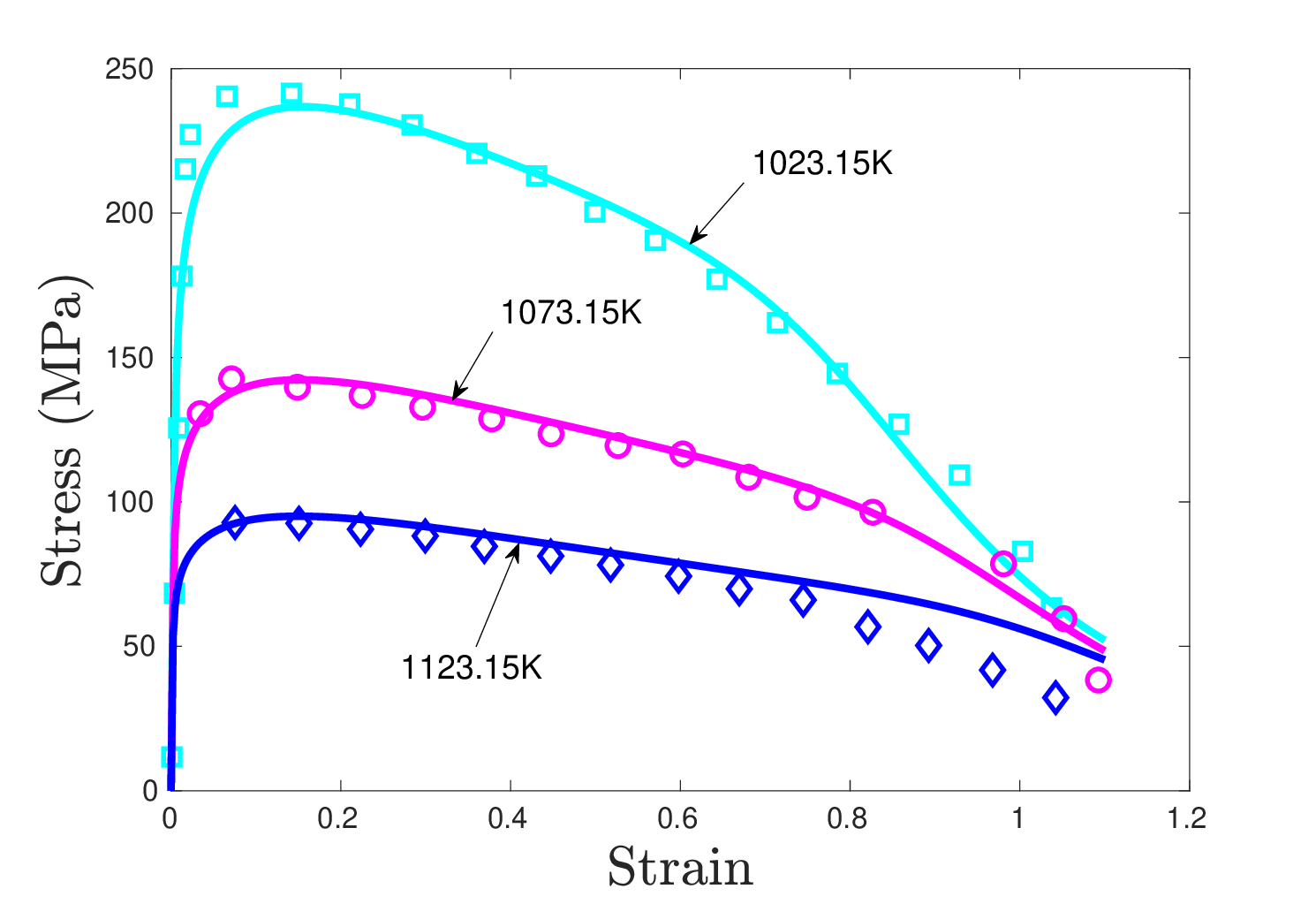}
\caption{}
\label{fig:TA15_0.001}
\end{subfigure}

\caption{Model predictions of stress-strain relationship and validation against experimental data (symbols) for titanium alloy (TA15) at different strain rates and temperatures; (a) initial temperature = 1023.15 K,  (b) strain rate = 0.01/s, (c) strain rate = 0.001/s.}
\label{fig:Uniaxial_TA15}
\end{figure}

{\renewcommand{\arraystretch}{1.0}
\begin{table}[th]
\centering
\caption{Material constants for Fe370 mild steel.} 
\begin{tabular}{l c c c}
 \hline \hline			
Parameters & Values  &  Parameters &  Values\\
\hline	
 $\mu$ \; (GPa)  &  80.76     & $\nu$         &  0.3 \\
 $\rho$ \; (kg/m$^3$)  &  7870      & n           &  0.527 \\
 $S_0$ \; (MPa)        & 260 & m           &  0.567\\ 
 $H_0$ \; (MPa)        & 550  & $r$  &  0.12\\ 
 $H_s$ \; (MPa)        & 300 & $K_{Ic}$ (MPa $\sqrt{\text{m}})$  &  $75e6$\\ 
 $k_\gamma$ \; (kg/m)  & 280e6  & $\bar{k}_t$ $(\text{J/m})$ & $17e{3}$ \\
  $\mathcal{C}_v$ \; $(\text{J{kg}}^{-1}\text{K}^{-1})$& 460  &    $\bar{k}_p$         &  $5.62e{-9}$ \\
 $\theta_{ref}$ \; (K) & 100   &     $l_0$ (m)        &  $1e{-5}$ \\
 $\theta_{melt}$ \; (K)& 1750   &   $\dot{\gamma}_0$  $(s^{-1})$ & 10   \\
  \hline  
\end{tabular}
\label{table:Material_Parameters_Fe370}
\end{table}}

\subsection{Strain rate locking effect (Fe370)}

We now demonstrate strain rate locking in the ductile damage response of Fe370 mild steel observed at very high strain rates. Model predictions are validated against the experimental data reported by  \cite{mirone2019locking}. In the split Hopkinson tension bar (SHTB) test, as the strain rate in the necking region increases manifold, the material response becomes insensitive to further variations in the strain rate. The material parameters used in the numerical simulations are given in Table \ref{table:Material_Parameters_Fe370}. We consider the following exponential form of the temporal function  $\mathcal{Z}$ (see Eq. \ref{eqD:specialization of temporal function}), $\mathcal{Z} = k_{\rho} + k_{\gamma}  e^{ f\left(\mathcal{J}\right)}$ to model strain rate locking effect, where $\mathcal{J} = \text{log}_{10}(\Delta \dot{J}_0)$ and $\Delta \dot{J}_0$ is the rate of change in volume during $t$ and $t +\Delta t$.  $\zeta_\gamma $ is given as (see Eq. \ref{eqD:micro-inertia}),

\begin{equation}
\zeta_\gamma = 2k_{\rho} + 2 {k}_\gamma  e^{ f\left(\mathcal{J}\right)}
\label{eq:zeta_LockedD}
\end{equation}

Stress-strain response of Fe370 mild steel predicted from the present model is shown in Fig. \ref{fig:Fe370_mild_steel_D} which shows good agreement with the experimental observation reported in \cite{mirone2019locking}. Stress-strain response of Fe370 mild steel shows strain rate sensitivity in the low strain rate regime (see Fig. \ref{fig:strainRLD2}, \ref{fig:Fe370_strainRL_D2}). However, for sufficiently high strain rates, the response is mostly insensitive to further variations in the strain rate; (see  Fig. \ref{fig:strainRLD1}, \ref{fig:Fe370_strainRL_D1S1}). In the higher strain rate regime, the nonlinear micro-inertia (see, Eq. \ref{eq:zeta_LockedD}) due to defect motion becomes insensitive to strain rate variation and attains about the same value for different strain rates; see Fig. \ref{fig:Fe370_zeta_gammaP_strainR}. It is important to note that the micro-inertia term given in Eq. \ref{eq:zeta_LockedD} depends on the rate of change in volume through the function $f(\mathcal{J}) $, which may be very significant for a dilatational plasticity model.

\begin{figure}
\begin{subfigure}{0.45\textwidth}
\centering
\includegraphics[width=0.95\linewidth]{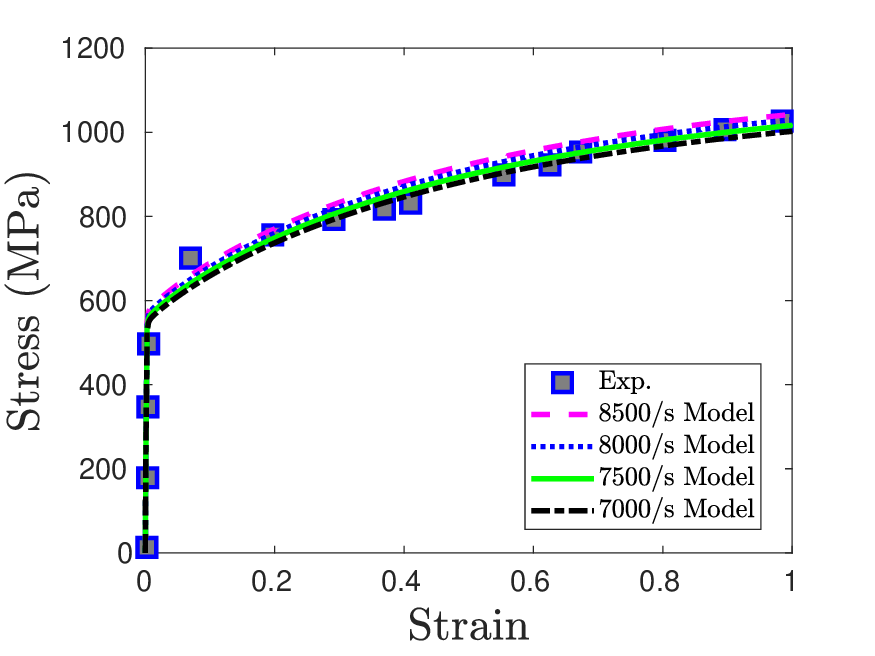}
\caption{}
\label{fig:strainRLD1}
\end{subfigure}
\hfill
\begin{subfigure}{0.45\textwidth}
\centering
\includegraphics[width=0.95\linewidth]{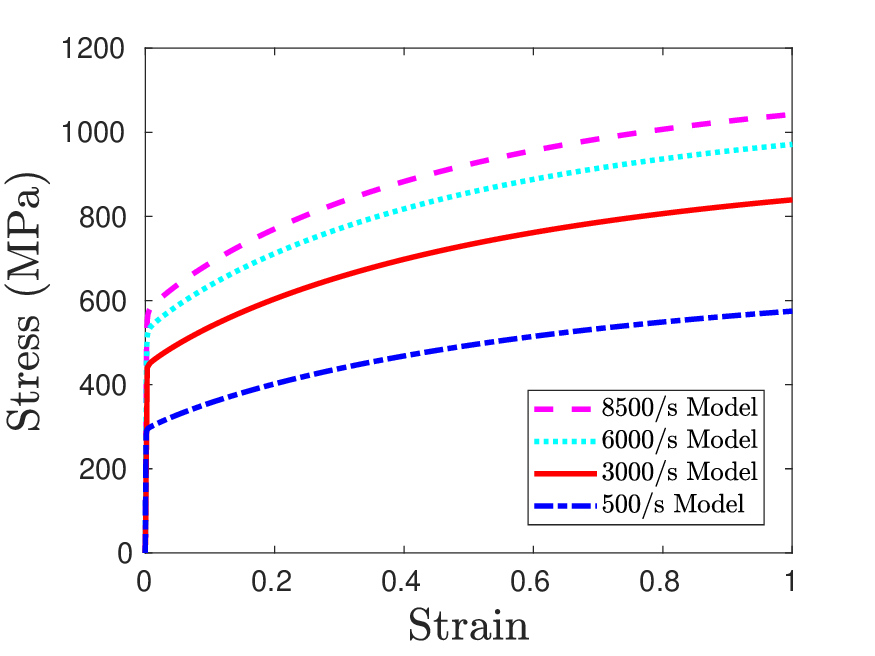}
\caption{}
\label{fig:strainRLD2}
\end{subfigure}

\vspace{1cm}
\begin{subfigure}{0.45\textwidth}
\centering
\includegraphics[width=0.95\linewidth]{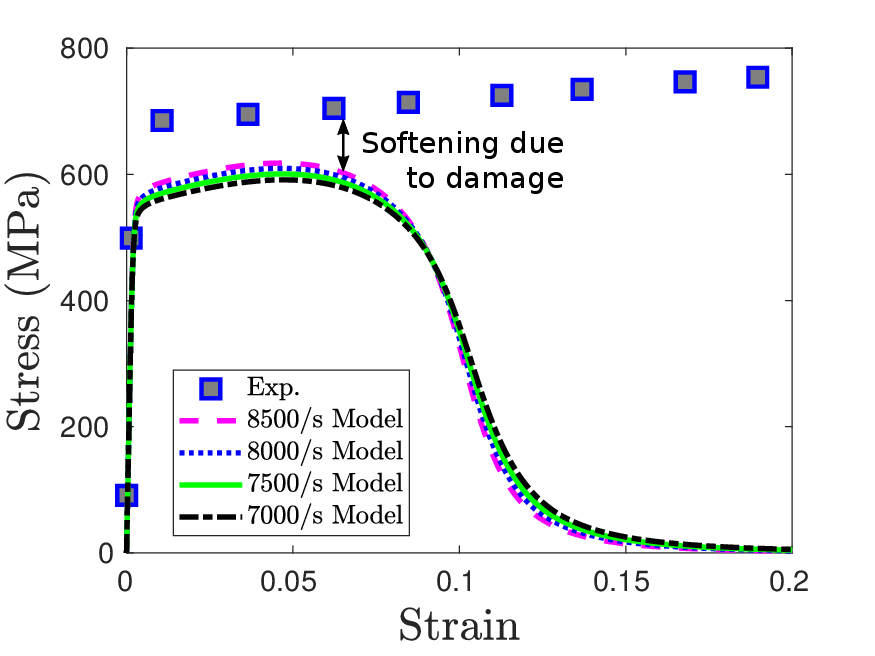}
\caption{}
\label{fig:Fe370_strainRL_D1S1}
\end{subfigure}
\hfill
\begin{subfigure}{0.45\textwidth}
\centering
\includegraphics[width=0.95\linewidth]{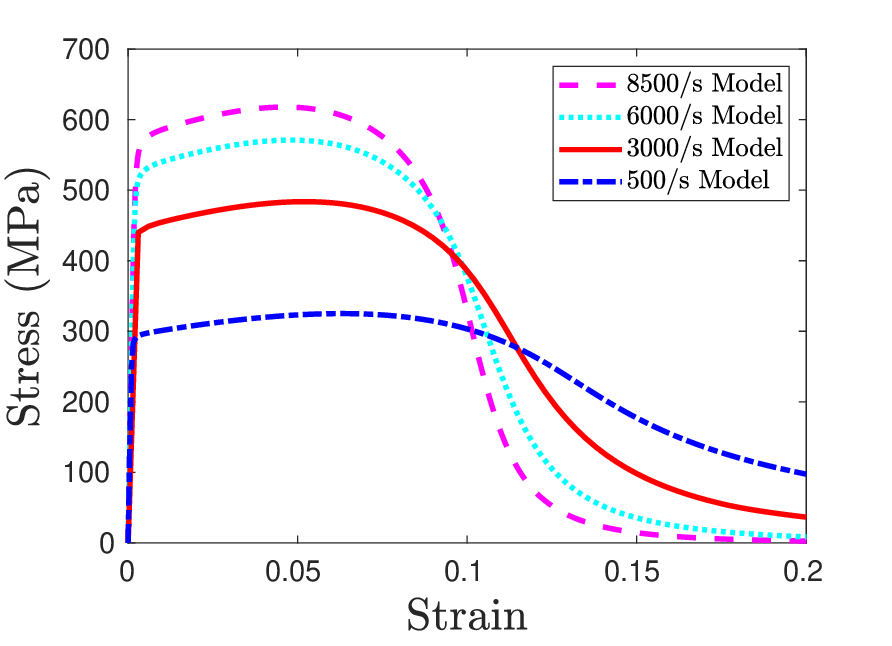}
\caption{}
\label{fig:Fe370_strainRL_D2}
\end{subfigure}

\begin{subfigure}{0.98\textwidth}
\centering
\includegraphics[width=0.45\linewidth]{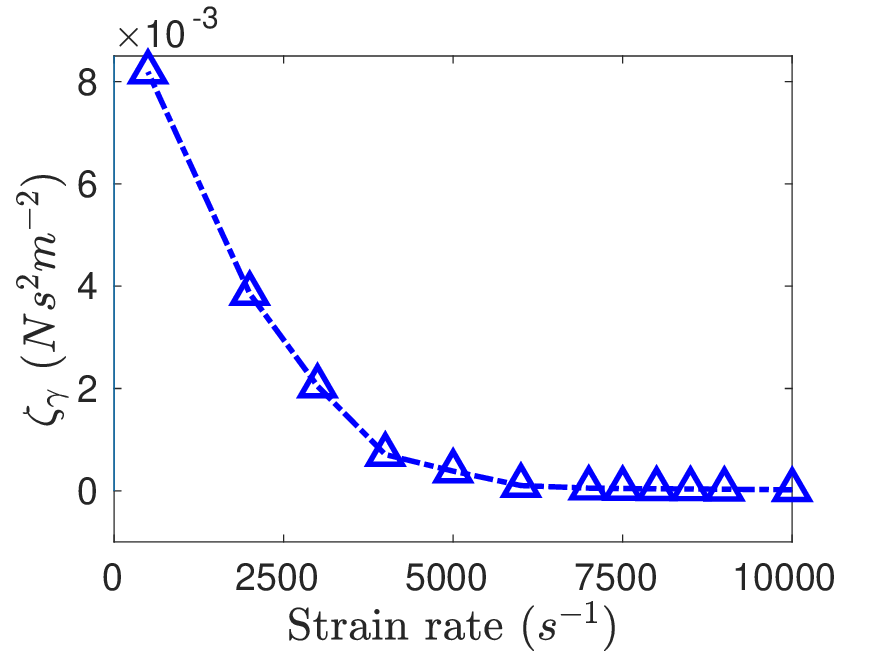}
\caption{}
\label{fig:Fe370_zeta_gammaP_strainR}
\end{subfigure}
\caption{Strain rate locking effect in the stress-strain curve of Fe370 mild steel at room
temperature. Response curve shows sensitivity in the lower strain rate regime; but is mostly
insensitive to further strain rate variations at higher strain rates. Figs. \ref{fig:strainRLD1} and \ref{fig:strainRLD2} represent stress-strain plots without accounting for damage while Figs. \ref{fig:Fe370_strainRL_D1S1} and \ref{fig:Fe370_strainRL_D2} show stress-strain plots with damage. $\zeta_\gamma$ denotes the coefficient of micro-inertia due to defect/dislocation motion.}
\label{fig:Fe370_mild_steel_D}
\end{figure}


\begin{figure}[h]
\centering
\includegraphics[scale=.28]{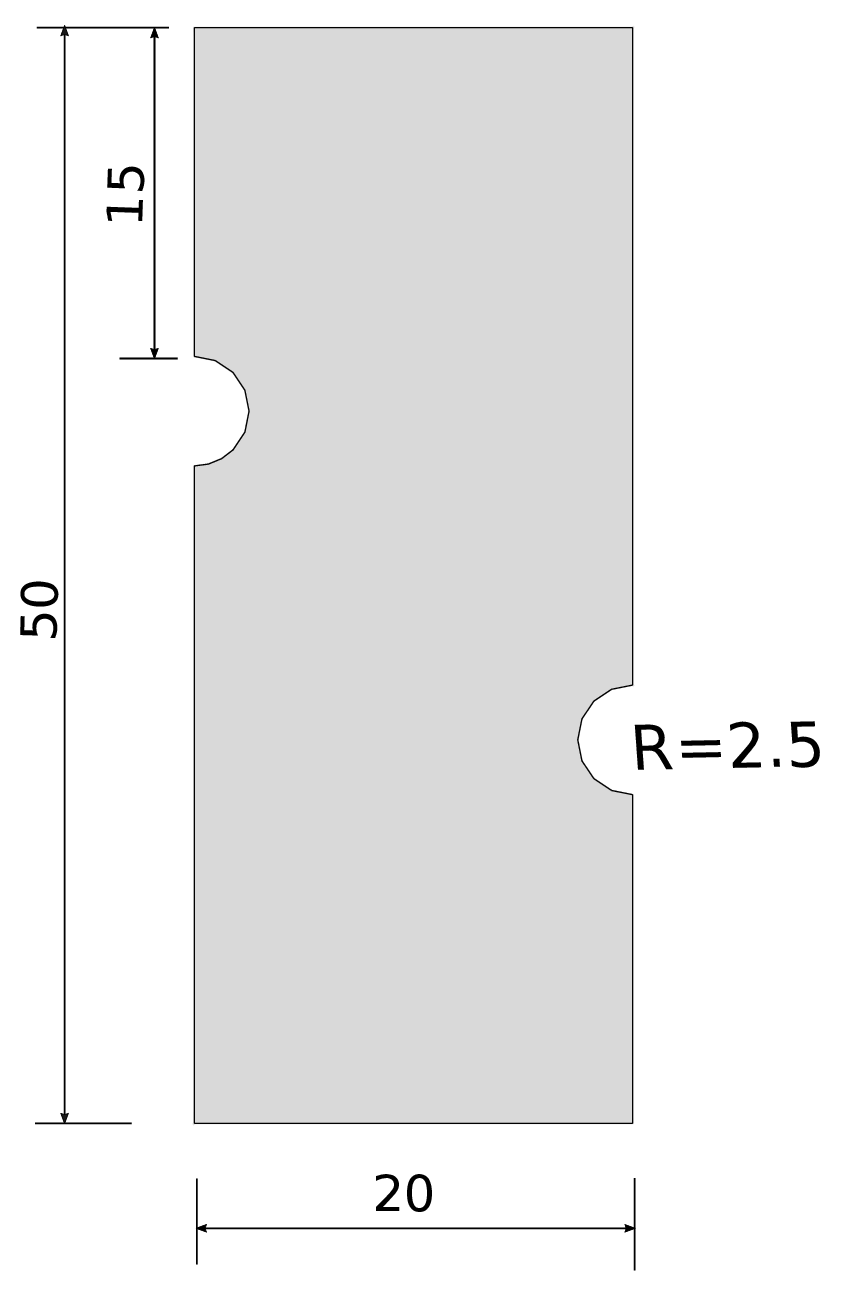}
\caption{Schematic of an asymmetrically notched plane strain specimen. Dimensions are in mm.}
\label{fig:schematic_asymmetrically_notch}
\end{figure}

\subsection{Asymmetrically notched specimen under tension}

\begin{figure}
\begin{subfigure}{0.45\textwidth}
\centering
\includegraphics[width=0.90\linewidth]{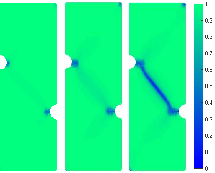}
\caption{}
\label{fig:Asy_phi}
\end{subfigure}
\hfill
\begin{subfigure}{0.45\textwidth}
\centering
\includegraphics[width=0.90\linewidth]{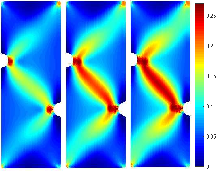}
\caption{}
\label{fig:Asy_gammaP}
\end{subfigure}

\begin{subfigure}{0.45\textwidth}
\centering
\includegraphics[width=0.90\linewidth]{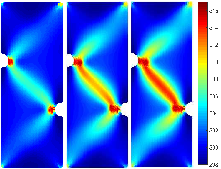}
\caption{}
\label{fig:Asy_Temp}
\end{subfigure}
\hfill
\begin{subfigure}{0.45\textwidth}
\centering
\includegraphics[width=0.90\linewidth]{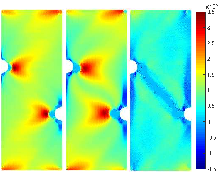}
\caption{}
\label{fig:Asy_Pressure}
\end{subfigure}

\caption{ Contour plots of field variables for asymmetrically notched specimen under tensile loading at times $90, 120$ and 200 $ \mu s $; (a) damage field ($\phi$), (b) equivalent plastic strain, (c) temperature (K), (d) hydrostatic pressure (Pa)} 
\label{fig:Asy_notched}
\end{figure}

This subsection demonstrates ductile damage in an asymmetrically notched plain strain specimen subjected to tensile loading. The schematic of the specimen is shown in Fig. \ref{fig:schematic_asymmetrically_notch}. We use the same material parameters of OFHC copper given in Table \ref{table:Material Parameters OFHC_TA15}. The top and bottom surfaces are subjected to a velocity of 16m/s in upward and downward directions, respectively. All other edges are considered traction-free. A total of 35284 particles are used to discretize the specimen in the numerical simulation, placed at a uniform spacing, $\Delta x= \Delta y= 1.667 \times 10^{-4}$m and horizon size of $1.05\Delta x$ is adopted. $\Delta x$ and $\Delta y$ respectively denote the particle spacings along horizontal and vertical directions. $l_\phi$ is taken as $2\Delta x$.  Fig. \ref{fig:Asy_phi} shows the contour plot of the damage field variable ($\phi$). Contour plots of equivalent plastic strain and temperature are presented in Figs. \ref{fig:Asy_gammaP} and   \ref{fig:Asy_Temp} respectively, whereas Fig. \ref{fig:Asy_Pressure} shows the hydrostatic pressure distribution at times 90, 120 and 200 $\mu s$. From the contour plots, we see that as the loading begins (90$\mu s$), damage, equivalent plastic strain, temperature and pressure get localized predominantly at the tip of both the notches. Damage, equivalent plastic strain and temperature further advance in a  direction perpendicular to the loading and finally at an inclination with loading direction. At the complete damaged state (200$\mu s$), significant pressure drop can be seen and eventually, pressure becomes zero at the fracture surface.  A convergence study of the damage field, equivalent plastic strain, and temperature with different number of particles (8821, 24512, 35284) used in the discretization of the specimen are also presented in Figs. \ref{fig:Asy_con_phi}, \ref{fig:Asy_con_gammaP} and \ref{fig:Asy_con_temp} respectively at time $t$ = 200 $\mu s$.

\begin{figure}
\begin{subfigure}{0.24\textwidth}
\centering
\includegraphics[width=0.60\linewidth]{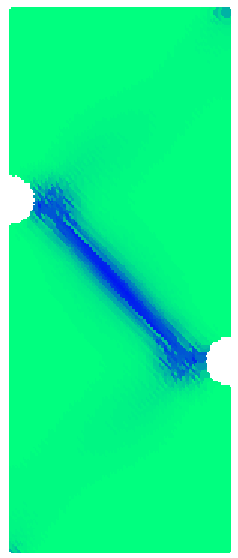}
\caption{}
\label{fig:Asy_con_phi_8821}
\end{subfigure}
\hfill
\begin{subfigure}{0.24\textwidth}
\centering
\includegraphics[width=0.60\linewidth]{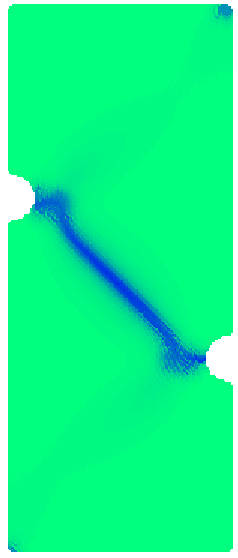}
\caption{}
\label{fig:Asy_con_phi_24512}
\end{subfigure}
\hfill
\begin{subfigure}{0.24\textwidth}
\centering
\includegraphics[width=0.60\linewidth]{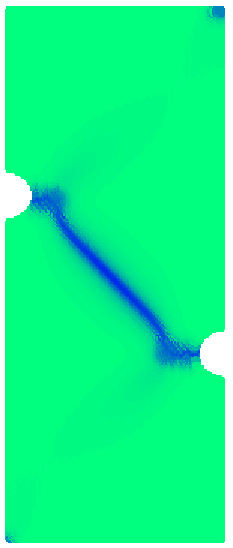}
\caption{}
\label{fig:Asy_con_phi_35284}
\end{subfigure}
\hspace{1mm}
\begin{subfigure}{0.07\textwidth}
\centering
\includegraphics[width=0.90\linewidth]{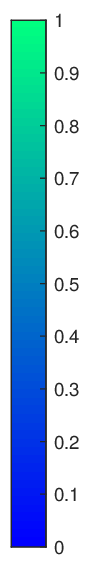}
\end{subfigure}

\caption{ Contour plots of the damage field ($\phi$) for convergence study at time 200 $\mu s $; (a) 8821 particles, (b) 24512 particles, (c) 35284 particles} 
\label{fig:Asy_con_phi}
\end{figure}
\begin{figure}
\begin{subfigure}{0.24\textwidth}
\centering
\includegraphics[width=0.60\linewidth]{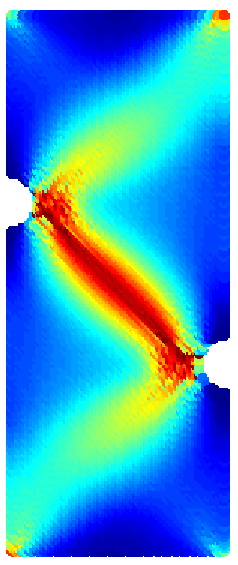}
\caption{}
\label{fig:Asy_con_gammaP_8821}
\end{subfigure}
\hfill
\begin{subfigure}{0.24\textwidth}
\centering
\includegraphics[width=0.60\linewidth]{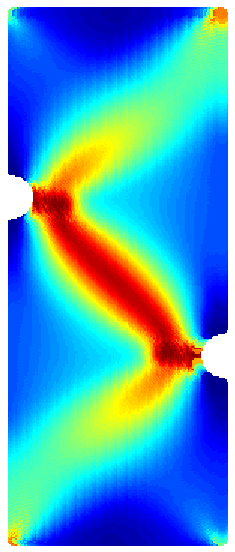}
\caption{}
\label{fig:Asy_con_gammaP_24512}
\end{subfigure}
\hfill
\begin{subfigure}{0.24\textwidth}
\centering
\includegraphics[width=0.60\linewidth]{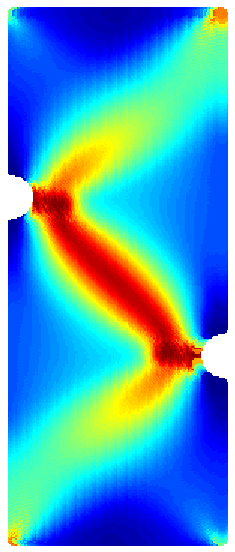}
\caption{}
\label{fig:Asy_con_gammaP_35284}
\end{subfigure}
\hspace{1mm}
\begin{subfigure}{0.07\textwidth}
\centering
\includegraphics[width=0.90\linewidth]{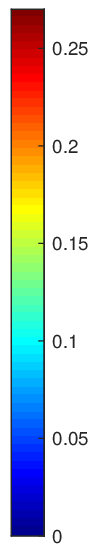}
\end{subfigure}

\caption{ Contour plots of equivalent plastic strain for convergence study at time 200 $\mu s $; (a) 8821 particles, (b) 24512 particles, (c) 35284 particles} 
\label{fig:Asy_con_gammaP}
\end{figure}
\begin{figure}
\begin{subfigure}{0.24\textwidth}
\centering
\includegraphics[width=0.60\linewidth]{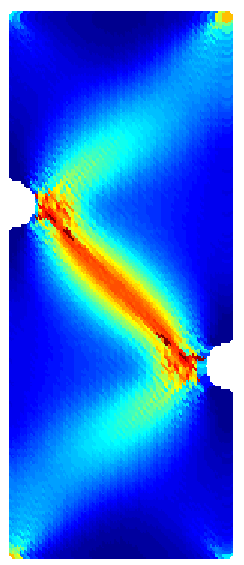}
\caption{}
\label{fig:Asy_con_temp_8821}
\end{subfigure}
\hfill
\begin{subfigure}{0.24\textwidth}
\centering
\includegraphics[width=0.60\linewidth]{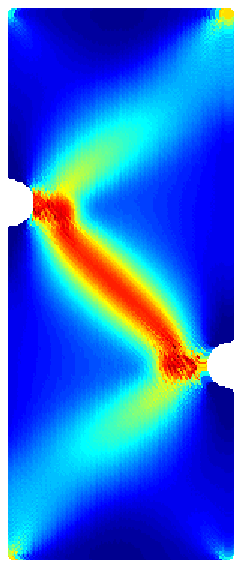}
\caption{}
\label{fig:Asy_con_temp_24512}
\end{subfigure}
\hfill
\begin{subfigure}{0.24\textwidth}
\centering
\includegraphics[width=0.60\linewidth]{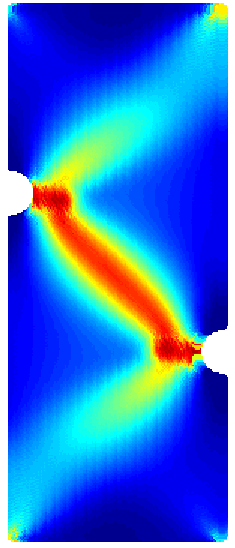}
\caption{}
\label{fig:Asy_con_temp_35284}
\end{subfigure}
\hspace{1mm}
\begin{subfigure}{0.07\textwidth}
\centering
\includegraphics[width=0.90\linewidth]{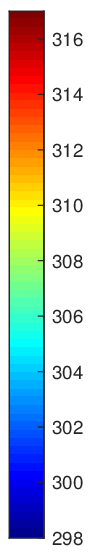}
\end{subfigure}

\caption{ Contour plots of temperature (K) for convergence study at time 200 $\mu s $; (a) 8821 particles, (b) 24512 particles, (c) 35284 particles} 
\label{fig:Asy_con_temp}
\end{figure}

\subsection{Cup-cone fracture in a dog-bone shaped round bar}

 \begin{figure}[h]
\centering
\includegraphics[scale=.35]{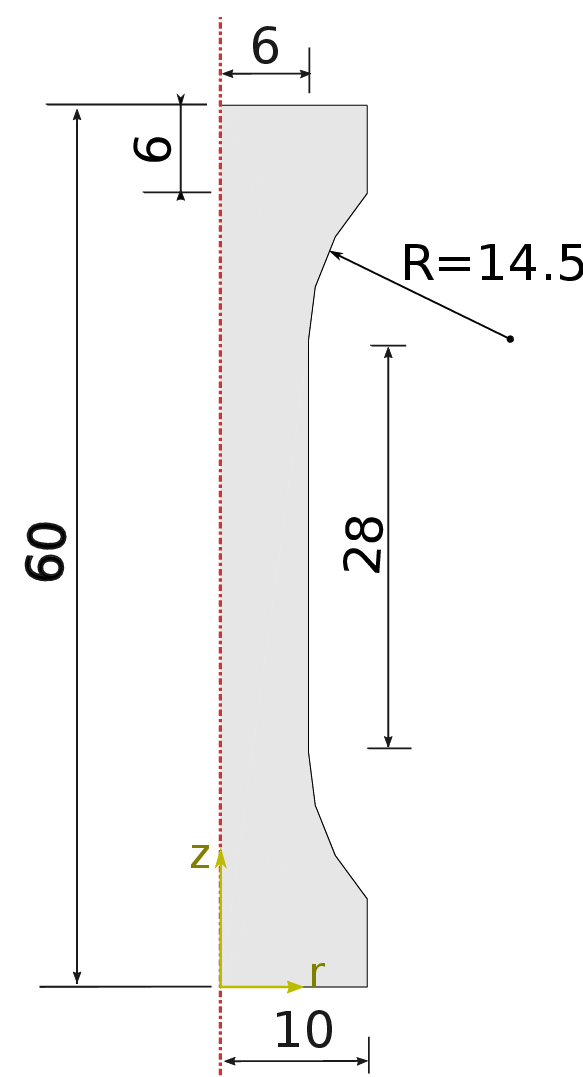}
\caption{Schematic of axi-symmetric section of a dog-bone shaped round bar. Dimensions are in mm.}
\label{fig:schematic_dog-bone}
\end{figure}

Now we present the cup-cone fracture in a dog-bone shaped round bar subjected to tensile loading.  We use the PD axisymmetric formulation (see Appendix \ref{axi_formulation}) for our numerical implementation. The schematic of the axisymmetric section is shown in Fig. \ref{fig:schematic_dog-bone}. The z-axis denotes the axis of symmetry. The same material parameters of OFHC copper given in Table \ref{table:Material Parameters OFHC_TA15} is used for numerical simulation. The specimen is loaded with an upward velocity of magnitude 16 m/s applied at the top edge whilst keeping the bottom edge restrained against displacement in the z-direction. All other edges are traction-free. Particles are placed at a uniform spacing along the z-direction, $\Delta z = 1 \times 10^{-4} $ and in the r-direction, $\Delta r = 1 \times 10^{-4} $. The horizon size is $1.05\Delta z$. $l_\phi$ is taken as twice of $\Delta r$.  Contour plots of cup-cone fracture at 330, 360 and 410 $(\mu s)$ are shown in Fig. \ref{fig:cup_cone}. Fig. \ref{fig:cup_phi} shows the contour plot of the damage field whereas contour plots of the equivalent plastic strain and temperature are shown in Figs. \ref{cup:cup_gammaP} and \ref{fig:cup_Temp}. Fig. \ref{fig:cup_cone} demonstrates qualitatively that the present model is able to predict the cup-cone fracture mechanism which consists of normal fracture in the  center and combined normal/shear at the specimen rim  \citep{tvergaard1984analysis,scheider2003simulation}. The normal crack diverts its path from radial to slant, implying a transition from tensile fracture to shear fracture.

\begin{figure}
\begin{subfigure}{0.3\textwidth}
\centering
\includegraphics[width=0.80\linewidth]{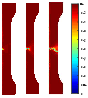}
\caption{}
\label{fig:cup_phi}
\end{subfigure}
\hfill
\begin{subfigure}{0.3\textwidth}
\centering
\includegraphics[width=0.80\linewidth]{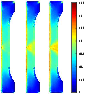}
\caption{}
\label{cup:cup_gammaP}
\end{subfigure}
\hfill
\begin{subfigure}{0.3\textwidth}
\centering
\includegraphics[width=0.80\linewidth]{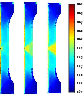}
\caption{}
\label{fig:cup_Temp}
\end{subfigure}

\caption{ Contour plots of field variables for cup-cone fracture at times $330, 360$ and 410 $\mu s $; (a) damage field ($\phi$) (b) equivalent plastic strain (c) temperature (K).} 
\label{fig:cup_cone}
\end{figure}


\section{Conclusions}
\label{vd:Conclusion}

This research article has dwelt on a rationally grounded ductile damage model for metals and alloys. The basis for this development is an exploitation of local translational and conformal symmetries implemented through a  space-time gauge theory. The evolutions of plastic flow and damage have been described using two gauge compensating field variables that emerge in the minimal replacement construct. The invariance of the energy density in local space-time translation and scaling is preserved through a space-time gauge covariant definition of partial derivatives. The fracture energy is constructed utilizing the gauge-invariant scalar curvature. This approach has also explicated on the geometric underpinnings of damage through local scaling of the metric, bypassing an internal variables paradigm used in classical modelling routes to ductile damage. The theory furnishes a thermodynamically consistent temperature evolution equation, and thus we eschew the use of the Taylor–Quinney coefficient. Apart from its scientific credibility, the model also allows a ready numerical implementation, with its predicting features showcased herein through several simulation results across different deformation scenarios. In the process, we have been able to reproduce certain experimental observations, e.g. the strain rate locking effect, ductile damage response of titanium alloy (TA15), etc. As an application, a few benchmark problems of ductile damage -- ductile fracture in an asymmetrically notched specimen and cup-cone fracture in the dog-bone shaped round bar, have been explored numerically. Besides extending the present theory to a unified electro-magneto-mechanical damage model, we also intend to explore the local rotational symmetry breaking for ductile fracture problems at extremely high strain rates, i.e. in the shock wave regime where the effect of lattice orientation could be significant.


\section*{Acknowledgments}
The partial financial support received from Indian Space Research Organization (ISRO), Government of India is gratefully acknowledged. The author would like to thank Prof. Roy, Indian Institute of Science Bangalore for his insightful suggestions.

\begin{appendices}
\section{Axisymmetric formulation for non-ordinary state-based peridynamics}
\label{axi_formulation}

Briefly, we present the axisymmetric formulation for non-ordinary state-based peridynamics. We exploit the symmetry of the structure and the boundary conditions to reduce the problem spatial dimension (e.g., PD power balance, etc.) from 3D to 2D by performing integration along the azimuth direction. Using the cylindrical coordinate system the two-point force state $\barbelow{\tb{T}}$ can be expressed in the component form along the radial (r), azimuth ($\varphi$) and height (z) directions as follows:

\begin{equation}
\barbelow{\tb{T}} = \barbelow{\text{T}}_r\left(r,\varphi,z,r',\varphi',z'\right) \bs{e}_r + \barbelow{\text{T}}_\varphi\left(r,\varphi,z,r',\varphi',z'\right)\bs{e}_\varphi + \barbelow{\text{T}}_z\left(r,\varphi,z,r',\varphi',z'\right)\bs{e}_z
\end{equation}

Due to axisymmetric loading and the boundary conditions the components of the displacement and velocity vectors along the azimuth dirction will vanish in contrast to the two-point functions and the partial derivatives of any field variables with respect to $\varphi$ will also be zero. Similarly, the scalar fields $\phi, \gamma^p$ and $\theta$ are functions of r and z only. In r-z plane, the shape tesor is given as:

\begin{equation}
\overline{\tb{K}}_{2D}=\int_{\mathcal{H}\left(\bs{X}\right)} \omega\left(|\bs{\xi}_{2D}|\right) \left(\bs{\xi}_{2D} \otimes \bs{\xi}_{2D} \right) dr' dz'
\label{eq:shape_tensor_axi}
\end{equation}

The non-local deformation gradient may be given as:

{\renewcommand{\arraystretch}{2.0}
\begin{equation}
\overline{\tb{F}}\left(\barbelow{\tb{Y}}\right) = \begin{bmatrix}
\overline{\tb{F}}_{2D} &  \bs{0}^{\mathsf{T}} \\
\bs{0}  &   \frac{\text{y}_r}{X_r}
\end{bmatrix}   
\end{equation}}

The relative deformation vector is given as, $\barbelow{\tb{Y}}_{2D}=\left[\text{y}'_r-\text{y}_r \;\;  \text{y}'_z-\text{y}_z \right]^{\mathsf{T}}$ and $\bs{\xi}_{2D}=\left[X'_r- X_r \;\; X'_z-X_z \right]^{\mathsf{T}}$ denotes the bond vector in the r-z plane. Note that the $\varphi \varphi$-component of the deformation gradient is non-zero and given as, $\overline{\text{F}}_{\varphi \varphi} = \frac{\text{y}_r}{X_r}$ and  $\overline{\tb{F}}_{2D}$ is given as:

\begin{equation}
\overline{\tb{F}}_{2D}\left(\barbelow{\tb{Y}}_{2D}\right) = \left[\int_{\mathcal{H}\left(\bs{X}\right)}\!\!\! \omega\left(|\bs{\xi}_{2D}|\right) \left(\barbelow{\tb{Y}}_{2D}\left\langle\bs{\xi}_{2D}\right\rangle \otimes \bs{\xi}_{2D} \right) dr' dz' \right]\overline{\tb{K}}^{-1}_{2D}
\label{eq:nonlocal_deformation_gradient2d}
\end{equation}

In component form,
{\renewcommand{\arraystretch}{2.0}
\begin{equation}
\overline{\tb{F}}_{2D} = \begin{bmatrix}
\frac{\partial \text{y}_r}{\partial X_r} &  \frac{\partial \text{y}_r}{\partial X_z}  \\
\frac{\partial \text{y}_z}{\partial X_r}   &   \frac{\partial \text{y}_z}{\partial X_z} 
\end{bmatrix} 
\end{equation}}

Taking the dot product on the both sides of Eq. \ref{eq:linear_momentum_PD} with $\dot{\tb{y}}$ and intergrating it over a finite sub-region $\mathcal{P}_t \subset \mathcal{B}_r$, we get:

\begin{equation}
\int_{\mathcal{P}_t} \rho\left(\bs{X}\right)\ddot{\tb{y}}\cdot \dot{\tb{y}} dV_{\bs{X}} = \int_{\mathcal{P}_t} \left[ \int_{\mathcal{B}_r} \left(\barbelow{\tb{T}}\left[\bs{X},t\right]\left\langle\bs{\xi}\right\rangle-\barbelow{\tb{T}}\left[\bs{X}',t\right]\left\langle-\bs{\xi}\right\rangle \right) dV_{\bs{X}'} + \tb{b}_0\left(\bs{X},t\right) \right] \cdot \dot{\tb{y}} dV_{\bs{X}}
\label{eq:linear_momentum_PD_axi1}
\end{equation}

Note that the horizon domain is extended to $\mathcal{B}_r$. Using the identity  and the antisymmetry of the integrand given in Eqs. \ref{eq:force_state_identity} and \ref{eq:identity_1} respectively, we may recast Eq. \ref{eq:linear_momentum_PD_axi1} as:

\begin{align}
\nonumber
2\pi \int_{z_1}^{z_2}\!\!\! \int_{r_1}^{r_2}\!\! \rho \left(\ddot{\text{y}}_r \dot{\text{y}}_r + \ddot{\text{y}}_z \dot{\text{y}}_z\right) rdr dz = \int_{z_1}^{z_2}\!\!\! \int_{r_1}^{r_2}\!\!\!\int_{z_1}^{z_2}\!\!\! \int_{r_1}^{r_2} &\left[ \left(\barbelow{\tilde{\tb{T}}}_2-\barbelow{\tilde{\tb{T}}}'_1\right) \dot{\text{y}}_r + \left(\barbelow{\hat{\tb{T}}}_2-\barbelow{\hat{\tb{T}}}'_1\right)\dot{\text{y}}_z \right] dr' dz' rdr dz \\
&+ 2\pi \int_{z_1}^{z_2}\!\!\! \int_{r_1}^{r_2}\left(\tb{b}_{0r} \dot{\text{y}}_r + \tb{b}_{0z} \dot{\text{y}}_z \right) rdr dz
\label{eq:linear_momentum_PD_axi2}
\end{align} 

where,

\begin{equation}\groupequation{
\begin{split}
&\barbelow{\tilde{\tb{T}}}_1 = \int_{0}^{2\pi}\!\!\! \int_{0}^{2\pi}\barbelow{\tb{T}}\cdot \bs{e}'_r r'd\varphi' d\varphi\\
&\barbelow{\tilde{\tb{T}}}_2 = \int_{0}^{2\pi}\!\!\! \int_{0}^{2\pi}\barbelow{\tb{T}}\cdot \bs{e}_r r'd\varphi' d\varphi\\
&\barbelow{\hat{\tb{T}}}_1 = \int_{0}^{2\pi}\!\!\! \int_{0}^{2\pi}\barbelow{\tb{T}}\cdot \bs{e}'_z r'd\varphi' d\varphi\\
&\barbelow{\hat{\tb{T}}}_2 = \int_{0}^{2\pi}\!\!\! \int_{0}^{2\pi}\barbelow{\tb{T}}\cdot \bs{e}_z r'd\varphi' d\varphi
\end{split}
}\end{equation}

Allowing the interactions only within a finite neighborhood $\mathcal{H} \subset \lceil r_1,r_2 \rceil \times \lceil z_1,z_2 \rceil$ around $\bs{X}$ and  using the arbitrariness of $\dot{\text{y}}_r$ and $\dot{\text{y}}_z$, we may write the equations of motion as:

\begin{equation}
\rho \ddot{\text{y}}_r = \frac{1}{2\pi}\int_{\mathcal{H}\left(\bs{X}\right)}  \left(\barbelow{\tilde{\tb{T}}}_2-\barbelow{\tilde{\tb{T}}}'_1\right) dr' dz' +  \tb{b}_{0r} 
\end{equation}

\begin{equation}
\rho \ddot{\text{y}}_z = \frac{1}{2\pi}\int_{\mathcal{H}\left(\bs{X}\right)}  \left(\barbelow{\hat{\tb{T}}}_2-\barbelow{\hat{\tb{T}}}'_1\right) dr' dz' +  \tb{b}_{0z} 
\end{equation}

Silimarly, we may derive the governing equations for plastic deformation Eq.\ref{eq:micro_force_balance_plasticity_PD} and damage field Eq.\ref{eq:micro_force_balance_damage_PD} in r-z plane as:

\begin{equation}
\zeta_\phi \ddot{\phi}\left(\bs{X},t\right)=\frac{1}{2\pi} \int_{\mathcal{H}\left(\bs{X}\right)}
\left(\barbelow{\tilde{\xi}}_\phi-\barbelow{\tilde{\xi}}'_\phi \right) dr' dz' - {\pi}_\phi\left(\bs{X},t\right)
\label{eq:micro_force_balance_damage_PD_axi}
\end{equation}

\begin{equation}
\zeta_\gamma \ddot{\gamma}^p\left(\bs{X},t\right)=\frac{1}{2\pi} \int_{\mathcal{H}\left(\bs{X}\right)}
\left(\barbelow{\tilde{\xi}}_\gamma-\barbelow{\tilde{\xi}}'_\gamma\right) dr' dz' - \bar{\pi}_\gamma\left(\bs{X},t\right)
\label{eq:micro_force_balance_plasticity_PD_axi}
\end{equation}

Next we present the constitutive correspondence with the classical material model in axisymmetric PD setup. Integrating along the azimuth direction $(\varphi)$, we may write Eq. \ref{eq:internal_energy_rate_local} as:

\begin{align}
\nonumber
\int_{z_1}^{z_2}\!\!\! \int_{r_1}^{r_2}\!\!\! \rho \dot{e}\, rdr dz = \frac{1}{2\pi}\int_{z_1}^{z_2}\!\!\! \int_{r_1}^{r_2}\!\!\! \int_{z_1}^{z_2}\!\!\! \int_{r_1}^{r_2} \bigg[\barbelow{\tilde{\tb{T}}}_1 \dot{\text{y}}'_r + \barbelow{\hat{\tb{T}}}_1 \dot{\text{y}}'_z -\barbelow{\tilde{\tb{T}}}_2 \dot{\text{y}}_r -\barbelow{\hat{\tb{T}}}_2 \dot{\text{y}}_z     
+  \barbelow{\tilde{\xi}}_\gamma \left( \dot{\gamma^p}' - \dot{\gamma}^p \right) \\ \nonumber
+ \barbelow{\tilde{\xi}}_\phi \left( \dot{\phi}' - \dot{\phi} \right)\bigg]dr'dz' rdr dz   
+\int_{z_1}^{z_2}\!\!\! \int_{r_1}^{r_2}\! \left[\left( \zeta_\gamma  \, \ddot{\gamma}^p + \bar{\pi}_\gamma \right) \dot{\gamma}^p   +  \left( \zeta_\phi \, \ddot{\phi} + {\pi}_\phi \right)\dot{\phi}  + \rho h \right] rdr dz  \\ 
- \frac{1}{2\pi}\int_{z_1}^{z_2}\!\!\! \int_{r_1}^{r_2}\!\!\! \int_{z_1}^{z_2}\!\!\! \int_{r_1}^{r_2}\!\! \barbelow{\tilde{q}}\, dr'dz'rdrdz
\label{eq:internal_energy_rate_local_axi}
\end{align}

where, we have used the following definitions:
\begin{equation}
\barbelow{\tilde{q}} = \int_{0}^{2\pi}\!\!\! \int_{0}^{2\pi}\left(\barbelow{q}- \barbelow{q}' \right)  r'd\varphi' d\varphi
\end{equation}

\begin{equation}
\barbelow{\tilde{\xi}}_\gamma = \int_{0}^{2\pi}\!\!\! \int_{0}^{2\pi} \barbelow{\xi}_\gamma  r'd\varphi' d\varphi
\end{equation}

\begin{equation}
\barbelow{\tilde{\xi}}_\phi = \int_{0}^{2\pi}\!\!\! \int_{0}^{2\pi}\barbelow{\xi}_\phi  r'd\varphi' d\varphi
\end{equation}

Now using the non-local approximation of the classical gradients in r-z plane and performing integration along azimuth direction, we may recast Eq. \ref{eq:nonlocal_internal_power} as:

\begin{align}
\nonumber
2\pi\int_{z_1}^{z_2}\!\!\! \int_{r_1}^{r_2}\!\!\! \rho \dot{e}\, rdr dz = \int_{\mathcal{B}_r} \left[ \overline{\tb{T}}: \dot{\overline{\tb{F}}}  + \overline{\bs{\xi}}_\gamma \cdot \dot{\overline{\tb{G}}}_{{\gamma}} + \overline{\bs{\xi}}_\phi \cdot \dot{\overline{\tb{G}}}_\phi \right] dV_{\bs{X}} + 2\pi\int_{z_1}^{z_2}\!\!\! \int_{r_1}^{r_2}\!\bigg[ \left( \zeta_\gamma  \, \ddot{\gamma}^p +\bar{\pi}_\gamma  \right) \dot{\gamma}^p \\ + \left( \zeta_\phi \, \ddot{\phi}  + {\pi}_\phi \right) \dot{\phi}  + \rho_0 h \bigg] rdr dz 
-\int_{z_1}^{z_2}\!\!\! \int_{r_1}^{r_2}\!\!\! \int_{z_1}^{z_2}\!\!\! \int_{r_1}^{r_2}\!\! \barbelow{\tilde{q}}\, dr'dz'rdrdz
\label{eq:nonlocal_internal_power_axi}
\end{align}

where,
\begin{equation}
\int_{\mathcal{B}_r} \overline{\tb{T}}: \dot{\overline{\tb{F}}} dV_{\bs{X}} = 2\pi \int_{z_1}^{z_2}\!\!\! \int_{r_1}^{r_2}\! \left( \int_{\mathcal{H}\left(\bs{X}\right)}\!\!\! \omega \overline{\tb{T}}_{2D}\overline{\tb{K}}^{-1}_{2D}\bs{\xi}_{2D} \cdot \dot{\barbelow{\tb{Y}}}_{2D}   \, dr'dz'  + \overline{\tb{T}}_{\varphi \varphi} \dot{\overline{\text{F}}}_{\varphi \varphi} \right) rdr dz
\label{eq:stress_power_PD_axi1}
\end{equation}

\begin{equation}
\int_{\mathcal{B}_r} \overline{\bs{\xi}}_\gamma \cdot \dot{\overline{\tb{G}}}_{{\gamma}} dV_{\bs{X}} = 2\pi \int_{z_1}^{z_2}\!\!\! \int_{r_1}^{r_2}\!  \int_{\mathcal{H}\left(\bs{X}\right)} \left[\omega \left(\overline{\bs{\xi}}_\gamma\right)_{2D} \cdot \bs{\xi}_{2D} \overline{\tb{K}}_{2D}^{-1} \right]  \dot{\barbelow{\Gamma}}^p_{2D} \,  dr' dz'  rdr dz
\label{eq:plast_power_PD_axi1}
\end{equation}
 
\begin{equation}
\int_{\mathcal{B}_r} \overline{\bs{\xi}}_\phi \cdot \dot{\overline{\tb{G}}}_{\phi} dV_{\bs{X}} = 2\pi \int_{z_1}^{z_2}\!\!\! \int_{r_1}^{r_2}\!  \int_{\mathcal{H}\left(\bs{X}\right)} \left[ \omega \left(\overline{\bs{\xi}}_\phi\right)_{2D} \cdot \bs{\xi}_{2D} \overline{\tb{K}}_{2D}^{-1} \right]  \dot{\barbelow{\Phi}}_{2D} \,  dr' dz'  rdr dz
\label{eq:damage_power_PD_axi1}
\end{equation}

We may extend the domain of integration from $\mathcal{H}$ to $\mathcal{B}_r$ in Eq.\ref{eq:stress_power_PD_axi1}-\ref{eq:damage_power_PD_axi1} and assuming $\overline{\tb{S}}_{2D} = \omega \overline{\tb{T}}_{2D}\overline{\tb{K}}^{-1}_{2D}\bs{\xi}_{2D}$, we may express Eq. \ref{eq:stress_power_PD_axi1} in component form as:

\begin{equation}
\int_{\mathcal{B}_r} \overline{\tb{T}}: \dot{\overline{\tb{F}}} dV_{\bs{X}} = 2\pi \int_{z_1}^{z_2}\!\!\! \int_{r_1}^{r_2}\left[ \int_{z_1}^{z_2} \!\!\! \int_{r_1}^{r_2}\! \left[ \left(\overline{\tb{S}}_{2D} \right)_r \left(\dot{\text{y}}'_r - \dot{\text{y}}_r \right) + \left(\overline{\tb{S}}_{2D} \right)_z \left(\dot{\text{y}}'_z - \dot{\text{y}}_z \right) \right] dr'dz'  + \overline{\tb{T}}_{\varphi \varphi} \frac{\dot{\text{y}}_r}{X_r} \right] rdr dz
\label{eq:stress_power_PD_axi}
\end{equation}

Now comparing Eqs.\ref{eq:internal_energy_rate_local_axi} and \ref{eq:nonlocal_internal_power_axi}, we get:

\begin{equation}\groupequation{
\begin{split}
&\barbelow{\tilde{\tb{T}}}_1 = 2\pi \left(\overline{\tb{S}}_{2D} \right)_r\\
&\barbelow{\tilde{\tb{T}}}_2 = 2\pi\left(\overline{\tb{S}}_{2D} \right)_r -  \frac{2\pi\overline{\tb{T}}_{\varphi \varphi}}{X_r \overline{\alpha}} \\
&\barbelow{\hat{\tb{T}}}_1 = \barbelow{\hat{\tb{T}}}_2 = 2\pi\left(\overline{\tb{S}}_{2D} \right)_z\\
&\barbelow{\tilde{\xi}}_\gamma = \omega \left(\overline{\bs{\xi}}_\gamma\right)_{2D} \cdot \bs{\xi}_{2D} \overline{\tb{K}}_{2D}^{-1} \\
&\barbelow{\tilde{\xi}}_\phi= \omega \left(\overline{\bs{\xi}}_\phi \right)_{2D} \cdot \bs{\xi}_{2D} \overline{\tb{K}}_{2D}^{-1}
\end{split}
}\end{equation}

where, $\overline{\alpha} = \int_{\mathcal{H}}  dr' dz'$ 
\end{appendices}


\bibliography{DuctileDamage_SpaceTime}
\end{document}